\documentclass[pdflatex,sn-mathphys-ay]{sn-jnl}% Math and Physical Sciences Numbered Reference Style

\usepackage{setspace}
\usepackage{latexsym}
\usepackage{amssymb} 
\usepackage{amsmath}
\usepackage{amsthm}
\usepackage{dsfont}
\usepackage{mathrsfs}
\usepackage{placeins}
\usepackage{bm}
\usepackage{subfig}
\usepackage{enumitem}
\usepackage{graphicx}
\usepackage[title]{appendix}%
\usepackage{hyperref}
\usepackage[nameinlink]{cleveref}
\hypersetup{
	colorlinks   = true, %Colours links instead of ugly boxes
	urlcolor     = blue, %Colour for external hyperlinks
	linkcolor    = blue, %Colour of internal links
	citecolor   = blue %Colour of citations
	urlcolor     = blue, %Colour for external hyperlinks
	linkcolor    = blue, %Colour of internal links
	citecolor   = blue %Colour of citations
}

\theoremstyle{thmstyleone}
\newtheorem{Theorem}{Theorem}[section]
\newtheorem{Corollary}{Corollary}[section]
\newtheorem{Proposition}{Proposition}[section]
\newtheorem{Lemma}{Lemma}[section]
\theoremstyle{thmstylethree}%
\newtheorem{Definition}{Definition}[section]

\newtheorem*{Assumptions}{Assumptions}
\theoremstyle{thmstyletwo}
\newtheorem{Example}{Example}[section]
\newtheorem{Remark}{Remark}[section]

\newcommand{\ind}{\mathds{1}}
\newcommand{\R}{\mathds{R}}
\newcommand{\N}{\mathds{N}}

\newcommand{\E}{\mathrm{E}}

\newcommand{\var}{{\mathrm{Var}}}

\newcommand{\wh}{\widehat}
\newcommand{\wt}{\widetilde}

\newcommand{\peq}{\phantom{=}}
\newcommand{\xix}{\bm X_{i, \bm x}}
\newcommand{\Xx}{\bm X_{\bm x}}

\newcommand{\bs}{\bm{s}}
\newcommand{\bt}{\bm{t}}

\newcommand{\bx}{\bm{x}}

\newcommand{\bz}{\bm{z}}

\newcommand{\bV}{\bm{V}}

\newcommand{\bX}{\bm{X}}
\newcommand{\bY}{\bm{Y}}

\newcommand{\Dcal}{\mathcal{D}}

\newcommand{\Fcal}{\mathcal{F}}
\newcommand{\Gcal}{\mathcal{G}}
\newcommand{\Hcal}{\mathcal{H}}

\newcommand{\Scal}{\mathcal{S}}

\DeclareMathOperator{\argmin}{\arg\,\min}

\newcommand*\colvec[3][]{
    \begin{pmatrix}\ifx\relax#1\relax\else#1\\\fi#2\\#3\end{pmatrix}
}

\numberwithin{equation}{section}
\allowdisplaybreaks
\usepackage{color}

\def\boxit#1{\vbox{\hrule\hbox{\vrule\kern6pt
          \vbox{\kern6pt#1\kern6pt}\kern6pt\vrule}\hrule}}
\begin{document}

	% Arabic numbering for actual thesis // reset counter to 1
	\setcounter{page}{1}
	\pagenumbering{arabic}  
	
	% title page
	
\title{On dimension reduction in conditional dependence models}


\author[1]{\fnm{Thomas} \sur{Nagler}}\email{t.nagler@lmu.de}

\author[2]{\fnm{Gerda} \sur{Claeskens}}\email{gerda.claeskens@kuleuven.be}

\author[3]{\fnm{Ir\`ene} \sur{Gijbels}}\email{irene.gijbels@kuleuven.be}

\affil[1]{\orgdiv{Department of Statistics}, \orgname{LMU Munich}, \orgaddress{\street{Ludwigstra\ss e 33}, \city{Munich}, \postcode{80539}, \country{Germany}}}

\affil[2]{\orgdiv{ORStat and Leuven Statistics Research Center}, \orgname{KU Leuven}, \orgaddress{\street{Naamsestraat 69, box 3555}, \city{Leuven}, \postcode{3000},  \country{Belgium}}}

\affil[3]{\orgdiv{Department of Mathematics and Leuven Statistics Research Center}, \orgname{KU Leuven}, \orgaddress{\street{Celestijnenlaan 200B, box 2400}, \city{Leuven}, \postcode{3001},  \country{Belgium}}}

\abstract{
Inference of the conditional dependence structure is challenging when many covariates are present. In numerous applications, only a low-dimensional projection of the covariates influences the conditional distribution. The smallest subspace that captures this effect is called the central subspace in the literature.
  We show that inference of the central subspace of a vector random variable $\bm Y$ conditioned on a vector of covariates $\bm X$ can be separated into inference of the marginal central subspaces of the components of $\bm Y$ conditioned on $\bm X$ and on the copula central subspace, that we define in this paper.
Further discussion addresses sufficient dimension reduction subspaces for conditional association measures. An adaptive nonparametric method is introduced for estimating the central dependence subspaces, achieving parametric convergence rates under mild conditions. Simulation studies illustrate the practical performance of the proposed approach.}


\maketitle

	%\doublespacing
    	
	% content 
	\section{Introduction}

Dependence among the components of a random vector $\bm Y$ of length $q$ is an important concept in statistics and many applied disciplines. Medical scientists analyze the dependence in chemical profiles of tumor cells \citep{Helmlinger1997}, economists study the dependence between economic indicators \citep{rodriguez2007}, hydrologists model the dependence between flood characteristics \citep{zhang2006}, etc.

Often one also wants to investigate or control for the effect of other variables
$\bm X = (X_1, \dots, X_p)^\top$ on the dependence in $\bm Y$. Continuing the
examples from the fist paragraph, the medical scientist may be interested how
the dependence is influenced by the patient's age or the time passed since
outbreak of cancer, the economists may want to control for a country's size and
demography when analyzing the dependence between economic indicators, and the
hydrologists could ask for the influence of geographical properties on the
dependence between flood characteristics.

Such questions can be answered by conditional dependence models.
%\textcolor{red}{(reference needed)}.
We approach this via modeling the conditional copula of $\bm Y$ given $\bm X$.
This is an active field of research. 
Several inference methods from conditional copula models have been investigated.
\citet{patton2001} considered fully parametric models where the copula's
parameter can be expressed as as a known, parametric function of the covariates.
More flexible semiparametric models were discussed by \citet{acar2011},
\citet{Abegaz2012}, \citet{vatter2015}, and \citet{fermanian2018single}. Here,
the copula parameter is allowed to exhibit an unknown, non- or semiparametric
relationship with the covariates. Fully nonparametric estimators of the
conditional copula and association measures were discussed by
\citet{veraverbeke2011} and \citet{gijbels2011}. In many of these contributions
the dimension $p$ of the covariate vector $\bm X$ is assumed to be small (often
$p = 1$). The reasons are diverse but evident: In parametric and semiparametric
models, structural assumptions such as linearity and additivity become more
restrictive when $p$ is large. Additionally, issues with identifiability and
numerical instability arise. For nonparametric methods, the \emph{curse of
dimensionality} deteriorates the accuracy of the estimators; a dimension $p \ge 5$
is often considered too large already.

A popular solution to issues with high dimensional covariates are techniques for
sufficient dimension reduction. The goal of such methods is to find a
projection of the covariate vector $\bm X$ onto a space of dimension $d \ll p$
such that all of the relevant information in $\bm X$ is preserved. In regression
problems, the theory of \emph{central subspaces} \citep{cook1994, cook1998,
cook2002} serves as the foundation of a broad variety of methods.
\citet{LiArtemiouLi2011} developed a principal support vector machine approach for use in linear and nonlinear sufficient dimension reduction. \citet{HaffendenArtemiou2024} used
the sliced inverse mean difference for dimension reduction
and \citet{ZhangXueLi2024} studied dimension reduction for Fr\'echet regression, to give just a few examples.

In this paper we study dimension reduction methods in conditional dependence models and find that this is closely related to dimension reduction in regression problems. The main definitions and properties of these are briefly revised in \hyperref[sec:background]{Section~\ref{sec:background}}, after introducing the copula notation in \hyperref[sec:notation-etc]{Section~\ref{sec:notation-etc}}.
Novel concepts, such as a central copula subspace, we define
in \hyperref[sec:sklar]{Section~\ref{sec:sklar}}
and we study some of its theoretical properties.
Further, we introduce and study sufficient dimension reduction in the context of conditional association measures and concordance measures in \hyperref[sec:etadr]{Section \ref*{sec:etadr}} and
\hyperref[subsec:differences]{Section \ref*{subsec:differences}}.
Estimation of the central subspaces is explained in 
\hyperref[sec:estimation]{Section \ref*{sec:estimation}}.
\hyperref[sec:sims]{Section~\ref{sec:sims}} contains simulation results. Longer proofs are collected in the Appendix.

\section{Notation, definitions and background on copulas} \label{sec:notation-etc}

The most general tool to statistically model the dependence in multivariate distributions is the \emph{copula}. According to Sklar's theorem \citep{sklar1959}, the joint distribution $F$ of the random variables $Y_1 \sim
F_1$, \dots, $Y_q \sim F_q$ can be expressed as
\begin{align} \label{eq:sklar}
 	F(\bm y) = C\bigl\{F_1(y_1), \dots,  F_q(y_q)\bigr\}, \qquad \mbox{for all }\bm y \in \R^q,
\end{align}
where the function $C\colon [0,1]^q \to [0,1]$ is called the copula of the random vector $\bm Y = (Y_1, \dots,  Y_q)^\top$.  By its definition, a copula is a joint distribution function of uniformly distributed random variables. If $Y_1, \dots, Y_q$ are continuous random variables, what we shall assume from here on, the copula $C$ is unique. More specifically, $C$ is the multivariate
distribution function of the random vector $\bm U = (U_1, \dots, U_q)^\top =
(F_1(Y_1), \dots,  F_q(Y_q))^\top$. Hence,
$ %\begin{align*}
	C(\bm u) = P\bigl(U_{1} \le u_1, \dots, U_{q} \le u_q\bigr), \quad \mbox{for all } \bm u \in [0,1]^q.
$ %\end{align*}
The reverse is also true, for arbitrary marginals $F_1, \dots, F_q$ and
copula $C$, \eqref{eq:sklar} is a valid joint distribution.
Sklar's theorem shows that any joint
distribution can be decomposed into the marginal distributions and the copula.
Hence, the copula captures all of the non-marginal behavior of the random vector, i.e., the dependence between the random variables.

%---------------------------------------------------------------------------------------------------%
%\subsection{Conditional dependence and the dimension reduction problem}
%\label{subsec:cond_dep}

\citet{patton2001} stated a version of Sklar's theorem that includes a
conditioning on $\bm X$. Let $F_{\bm Y \mid \bm X}(\bm y) = P(\bm Y \le \bm y
\vert \bm X)$ be the distribution of a random vector $\bm Y$ conditional on $\bm
X$. Similarly, denote $F_{Y_j \mid \bm X}(y_j)= P(Y_j \le y_j \vert \bm X)$, $j
= 1, \dots, q$ as the conditional marginal distributions. Note that, due to
their dependence on $\bm X$, the conditional joint and marginal distributions
are random functions. Then for all $\bm y \in \R^q$ and almost surely
(\emph{a.s.} from here on),
\begin{align*}
	F_{\bm Y \mid \bm X}(\bm y) = C_{\bm X}\bigl\{F_{Y_1 \mid \bm X}(y_1), \dots, F_{Y_q \mid \bm X}(y_q)\bigr\},
\end{align*}
for some (random) copula function $C_{\bm X}\colon [0,1]^q \to [0,1]$ which is
called the conditional copula associated with the random vector $\bm Y \vert \bm
X$. It is the conditional multivariate distribution function of the random vector
\begin{align*}
	\bm U_{\bm X} = (U_{1, \bm X}, \dots, U_{q, \bm X})^\top = \bigl(F_{Y_1 \mid \bm X}(Y_1), \dots, F_{Y_q \mid \bm X}(Y_q)\bigr)^\top,
\end{align*}
conditional on $\bm X$, i.e., for all $\bm u \in [0,1]^q$,
\begin{align*}
	C_{\bm X}(\bm u) = P\bigl(U_{1, \bm X} \le u_1, \dots, U_{q, \bm X} \le u_q \vert\bm X\bigr) \quad a.s.
\end{align*}

  \section{Background on sufficient dimension reduction in regression problems}
\label{sec:background}

First, we revise the needed definitions and concepts before introducing the novel copula dimension reduction methods in Section~\ref{sec:sklar}.

Consider a regression problem with (possibly multivariate) response $\bm Y \in \R^q$, $q \in \N$, and covariate vector $\bm X \in \R^p$. We are interested in the conditional distribution of $\bm Y$ given $\bm X$.
Suppose there exists a projection matrix $\bm B \in \R^{p \times d}$, $d \le p$,  such that for all $\bm y \in \R^q$,
$ %\begin{align*}
P(\bm Y \le \bm y \vert \bm X) = P(\bm Y \le \bm y \vert \bm B^\top \bm X) \quad a.s.,
$ %\end{align*}
or more concisely:
\begin{align}
	F_{\bm Y \mid \bm X} = F_{\bm B^\top \bm  X} \quad a.s.
    \label{eq:suffdr}
\end{align}
Then $\bm B^\top \bm X$ contains all the information of $\bm X$ that is relevant for predicting $\bm Y$. The $p \times p$ identity matrix trivially satisfies \eqref{eq:suffdr}, so a matrix $\bm B$ satisfying \eqref{eq:suffdr} always exists. When such a matrix exists for $d < p$, we can reduce the dimension of the covariate vector without losing any information about the conditional expectation. In other words, knowledge of the $d$-dimensional vector $\bm B^\top \bm X$ is \emph{sufficient} for predicting $Y$ and we speak of \emph{sufficient dimension reduction}.

If we can find such a matrix for $d \ll p$, we benefit in two ways. Inference of
the reduced regression problem becomes much easier than the original one. In
particular, the curse of dimensionality inherent in nonparametric regression
problems can be mitigated substantially. Further, if $q = 1$, plotting the
response variable $Y$ against the reduced predictor $\bm B^\top \bm X$ provides
a low-dimensional visualization of their relationship. The initial dimension
does not have to be large to make the reduction useful. Even when $p = 5$, the
curse of dimensionality drastically deteriorates the quality of nonparametric
estimates and the data is difficult to visualize. These problems disappear when
the dimension can be reduced to $d \le 2$ which is typically the goal.

%------------------------------------------------------------------------------------------%
\subsection{The central subspace in regression models}
\label{subsec:cs}

It is easy to see that the matrix $\bm B$ is not identifiable. In fact, $\bm B$ could be replaced by any matrix whose columns span the same space as the columns of $\bm B$. Hence, the goal is to identify $\mathcal{S}(\bm B)$, the column space of $\bm B$.
Whenever $\bm B$ satisfies \eqref{eq:suffdr}, the space $\mathcal{S}(\bm B)$ is called a \emph{dimension reduction subspace}. The ultimate goal is to find the smallest dimension reduction subspace, called \emph{central subspace}.

\begin{Definition}[Central subspace] Let  $\mathcal{\bm B}_{Y \vert \bm X} =
\{\bm B  \in \R^{p \times p} \colon F_{\bm Y \mid \bm X} = F_{\bm B^\top \bm  X}\; a.s.\}$. If
$\mathcal{S}_{Y \vert \bm X} = \bigcap_{\bm B \in \mathcal{\bm B}_{Y \vert \bm X}}
\mathcal{S}(\bm B) $ is a dimension reduction subspace, it is called
the {\bfseries central subspace}.
\end{Definition}

The set $\bigcap_{\bm B \in \mathcal{\bm B}_{Y \vert \bm X}} \mathcal{S}(\bm B)$
always contains at least the zero vector. But the central subspace may not
exist, namely when $\bigcap_{\bm B \in \mathcal{\bm B}_{Y \vert \bm X}}
\mathcal{S}(\bm B)$ is not a dimension reduction subspace. But existence can be
guaranteed under rather mild conditions; and whenever the central subspace
exists, it is also unique \citep[see,][]{chiaromonte1997foundations, cook1998}.
One simple condition is that the support of $\bm X$ is convex. More complex
conditions also applicable to discrete $\bm X$ can be found in
\citet{cook1998}.

Even if the central subspace does not exist, a dimension
reduction subspace always does. Hence, we can always find a matrix $\bm B$ that
satisfies \eqref{eq:suffdr}, which is what is most important in practice. But it
is possible that there are multiple dimension reduction subspaces whose
intersection is not a dimension reduction subspace. This phenomenon poses a
conceptual difficulty, but has only minor relevance in practice. We shall assume
throughout that central subspaces exist.

The following examples illustrate the new concepts for readers unfamiliar
with them.

\begin{Example}[Single index model]
\label{ex: SImodel}
Consider the single index model
\begin{align*}
Y =  g(\bm \beta^\top \bm X) + \sigma \epsilon, \quad\epsilon \sim \mathcal{N}(0, 1),
\end{align*}
where $\bm \beta \in \R^p$ and $\sigma > 0$ are model parameters
and $g\colon \R \to \R$ a possibly unknown function. The classical linear model
is recovered for $g(x) = x$. Equivalently, we can write $F_{Y \mid \bm X} = \Phi_{ g(\bm \beta^\top \bm X), \sigma^2}$, where $\Phi_{\mu, \sigma^2}$ denotes the normal distribution function with mean $\mu$ and variance $\sigma^2$.
 For any matrix $\bm B \in \R^{p \times d^\prime}$, $1 \le d^\prime \le p$, that has $\bm \beta$ in one of its columns, it holds $F_{\bm Y \mid \bm X} = F_{\bm B^\top \bm X}$ a.s.
Thus, $\mathcal{S}(\bm B)$ is a mean dimension reduction subspace. The intersection of all these spaces is $\mathcal{S}(\bm \beta)$, which is itself a dimension reduction subspace. Hence, $\mathcal{S}_{Y \vert \bm X} = \mathcal{S}(\bm \beta)$ is the central subspace. If $\bm \beta = \bm 0$, i.e., $Y = g(0) + \sigma \epsilon$, it holds $\mathcal{S}_{Y \vert \bm X} = \{\bm 0\}$ and the dimension of the
central subspace is $d = 0$. \qed
\end{Example}

\begin{Example}[Heteroscedastic multivariate model] \label{ex:mv}
Now consider
\begin{align*}
\bm Y =\bm g(\bm \beta^\top \bm X) + \bm A(\bm \gamma^\top \bm X)\bm \epsilon, \quad \bm \epsilon \sim \mathcal{N}(\bm 0, \bm I_{q \times q})
\end{align*}
where $\bm \beta, \bm \gamma \in \R^p$, $\bm g \colon \R \to \R^q$, 
$\bm A\colon \R \to \R^{q \times q}$, and $\bm I_{q \times q}$ denotes the identity matrix of dimension $q \times q$. The conditional distribution is  $F_{\bm Y
\mid \bm X} = \Phi_{\bm g(\bm \beta^\top \bm X), \bm A(\bm \gamma^\top \bm X)
\bm A(\bm \gamma^\top \bm X)^\top}$, where $\Phi_{\bm \mu, \bm \Sigma}$ is the
multivariate Gaussian cumulative distribution function. The corresponding central subspace is spanned
by the columns of $\bm B = (\bm \beta, \bm \gamma) \in \R^{p \times 2}$. \qed
\end{Example}

%----------------------------------------------------------------------------------------%

\subsection{The central mean subspace in regression models}
\label{subsec:cms}

In some cases we are only interested in the conditional mean, not the entire distribution. Suppose there exists a matrix $\bm B \in \R^{p \times d}$, $d \le p$,  such that
\begin{align}
	\E(\bm Y \vert \bm X) = \E(\bm Y \vert \bm B^\top \bm  X) \quad a.s.
    \label{eq:suffdr_mean}
\end{align}
For any matrix $\bm B$ that satisfies \eqref{eq:suffdr_mean}, $\mathcal{S}(\bm B)$ is called \emph{mean dimension reduction subspace} \citep{cook1994}. Again, the ultimate goal is to find the smallest mean dimension reduction subspace.
\begin{Definition}[Central mean subspace]
Let $\mathcal{B}_{\E(\bm Y \vert \bm X)} = \{\bm B  \in \R^{p \times p}\colon \bm \E(\bm Y \vert \bm X) = \E(\bm Y \vert \bm B^\top \bm  X) \; a.s\}.$ If $\mathcal{S}_{\E(\bm Y \vert \bm X)} = \bigcap_{\bm B \in \mathcal{\bm B}_{\E(\bm Y \vert \bm X)}} \mathcal{S}(\bm B)$ is a mean dimension reduction subspace, it is called the {\bfseries central mean subspace}.
\end{Definition}

Also the central mean subspace may not exist. But its existence can be
guaranteed under conditions similar to the ones for the central subspace. A
sufficient condition is that the support of $\bm X$ is convex
\citep[cf.,][]{cook1998,cook2002,cook2003}.

The central mean subspaces in Example \ref{ex: SImodel} is easily derived as
$\mathcal{S}_{\E(\bm Y\vert \bm X)} = \mathcal{S}(\bm \beta)$, which is equal to the
central subspace $\mathcal{S}_{\bm Y \vert \bm X}$. Equality holds because the
influence of $\bm X$ on $\bm Y$ is completely captured by the conditional mean. In
Example \ref{ex:mv}, the central mean subspace is also $\mathcal{S}_{\E(\bm Y\vert
\bm X)} = \mathcal{S}(\bm \beta)$. But now, $\mathcal{S}_{\E(\bm Y\vert \bm X)}
\neq \mathcal{S}_{Y\vert \bm X}$, because $\bm X$ also affects another other
aspect of the conditional distribution (the covariance). 
%The relationship between the two types of central subspaces is further discussed in the following.

%---------------------------------------------------------------------------------------------------%
\subsection{Connection between central and central mean subspaces}
\label{subsec:connection}

Since the conditional expectation is a functional of the conditional distribution, it is clear that $\mathcal{S}_{\E(\bm Y \vert \bm X)} \subseteq
\mathcal{S}_{\bm Y \vert \bm X}$. The intuitive consequence is that the dimension can possibly be reduced further when only the conditional mean is of concern. Under a location model (the covariate vector $\bm X$ affects the conditional distribution of $\bm Y$ given $\bm X$ only through the mean) the two spaces coincide.

More generally, it is possible to characterize the central subspace as the union of central mean subspaces indexed by a family of functions $\mathcal{G}$. This fact was noticed in a series of papers \citep{yin2002,zhang2006,zeng2010} and
summarized in a general setting by  \citet[Theorem 2.1]{yin2011} that for completeness we repeat here:

\begin{Proposition} \label{thm:cover}
Consider  a random vector $\bm Y \in \R^q$ and a family of functions $\mathcal{G}$. If
        \begin{enumerate}[label = (\roman*)]
		\item all $g \in \mathcal{G}$ are measurable maps from $\R^q$ to $\R$ and satisfy $\mathrm{Var}\{g(\bm Y)\} < \infty$,
        \item $\mathcal{G}$ is dense in the set $\{\ind(\bm Y \in A)\colon A \mbox{ is a Borel set in } \R^q\}$,
	\end{enumerate}
    it holds that $\mathcal{S}_{\bm Y \vert \bm X} = \mathrm{span}\{S_{\E\{g(\bm Y)\vert \bm X\}}, g \in \mathcal{G}\}.$
\end{Proposition}
Examples of families $\mathcal{G}$ complying with \autoref{thm:cover} are the
collection of indicator functions $\ind(\, \cdot \le \bm y)$, $\bm y \in \R^q$,
and the set of polynomials (provided all moments exist). \citet[Theorem
2.2]{yin2011} further show that, almost surely, a finite number of functions $g
\in \mathcal{G}$ is sufficient to cover $\mathcal{S}_{\bm Y \vert \bm X}$. An
important implication is: inference for the central subspace can be based on
inference of a finite number of central mean subspaces.

  \section{Copula decompositions in sufficient dimension reduction}
\label{sec:sklar}

Recall from \autoref{sec:notation-etc} that the conditional joint distribution
$F_{\bm Y \mid \bm X}$ of a random vector $\bm Y$ given $\bm X$ can be
decomposed into the conditional marginals $F_{1, \bm X}, \ldots, F_{q, \bm X}$
and the conditional copula $C_{\bm X}$. This suggests a similar decomposition
for the central subspace.

%---------------------------------------------------------------------------_%

\subsection{Copula decomposition of the dimension reduction subspaces}

\begin{Lemma} \label{lem:drsklar1}
Suppose a matrix $\bm B \in \R^{p \times d}$, $d \le p$, satisfies
\begin{align} \label{eq:FB}
F_{\bm Y \mid \bm X} = F_{\bm B^\top \bm X} \quad a.s.
\end{align}
Then there are matrices $\bm B_j \in \R^{p \times d_j}$, $j = 1, \dots, q$, and
$\bm B_C \in \R^{p \times d_C}$ such that  $\max\{d_1, \dots, d_q, d_C\} \le d$ that
are sub-matrices of $\bm B$ satisfying
\begin{gather} \label{eq:S12CB}
  \mathrm{span}\{\mathcal{S}(\bm B_1), \dots, \mathcal{S}(\bm B_q), \mathcal{S}(\bm B_C)\} = \mathcal{S}(\bm B),
\\
%\end{align}
%and
%\begin{align}  
\label{eq:12CB}
F_{Y_j \mid \bm X} = F_{Y_j \mid \bm B_j^\top \bm X}, \quad \mbox{for } j = 1, \dots, q, \quad \mbox{and}
\quad C_{\bm X} = C_{\bm B_C^\top \bm X} \quad a.s.
%\end{align}
\end{gather}
\end{Lemma}
\begin{proof}
Set $\bm B_1 = \dots = \bm B_j = \bm B_C = \bm B$, which trivially satisfies
\eqref{eq:S12CB}. Further note that
\begin{align*}
F_{Y_j \mid \bm X}(\cdot) = F_{\bm Y \mid \bm X}(\infty, \dots, \infty, \cdot,
\infty, \dots \infty),
\end{align*}
where $\infty$ is in all but the $j$th position, and
$C_{\bm X} = F_{\bm Y \mid \bm X}\{F_{Y_1 \mid \bm X}^{-1}, \dots,  F_{Y_q \mid
\bm X}^{-1}\}$. Thus, \eqref{eq:FB} implies \eqref{eq:12CB}.
\end{proof}
Lemma \ref{lem:drsklar1} implies that any dimension reduction subspace for the joint
distribution can be decomposed into dimension reduction subspaces for the
marginals and for the copula. The  decomposition is generally not disjoint; dimension
reduction subspaces for marginals and the copula can overlap. The next result
reverses \autoref{lem:drsklar1}: the span of dimension reduction subspaces for
marginals and the copula is a dimension reduction subspace for the joint
distribution.

\begin{Lemma} \label{lem:drsklar2}
Suppose there are matrices $\bm B_j \in \R^{p \times d_j}$, $j = 1, \dots, q$, and
$\bm B_C \in \R^{p \times d_C}$ with $\max\{d_1, \dots, d_q, d_C\} \le p$ satisfying \eqref{eq:12CB}.
Then there is a matrix $\bm B \in \R^{p \times d}$, $d \le \max\{ \sum_{j = 1}^q d_j + d_C, p\}$, that satisfies \eqref{eq:S12CB} and \eqref{eq:FB}.
\end{Lemma}
\begin{proof}
Set $\bar{\bm  B} = (\bm B_1, \dots, \bm B_q, \bm B_C) \in \R^{p \times (\sum_j d_j + d_C)}$. If $\sum_j d_j + d_C \le p$, set $\bm B = \bar{\bm  B}$, otherwise set $\bm B$ to a matrix derived from $\bar{\bm B}$ by eliminating linearly dependent columns until $\bm B \in \R^{p \times p}$. For this construction, there are full rank matrices $\bm A_j$, $j = 1, \dots, q$, $\bm A_C \in \R^{p \times p}$ such that $\bm A_j \bm B = (\bm B_j, \ldots)$, and $\bm A_C \bm B = (\bm B_C, \ldots)$. Then the claim follows from
$ %\begin{align*}
	F_{\bm Y \mid \bm X}(\bm y) = C_{\bm X}\bigl\{F_{Y_1 \mid \bm X}(y_1), \dots, F_{Y_q \mid \bm X}(y_q)\bigr\} \quad a.s.,
$ %\end{align*}
and the fact that conditional probabilities are invariant to full rank
transformations of the conditioning vector.
\end{proof}

%---------------------------------------------------------------------------_%

\subsection{The central copula subspace}

Lemmas~\ref{lem:drsklar1} and \ref{lem:drsklar2} can be seen as the dimension reduction version of Sklar's theorem. A natural question is whether similar results can be derived for central subspaces. Paralleling the definitions in \autoref{subsec:cs}, we define a \emph{copula dimension reduction subspace} as the column space $\mathcal{S}(\bm B)$ of a matrix $\bm B \in \R^{p \times d}$, $d \le p$, that satisfies
\begin{align} \label{eq:CB}
C_{\bm X} = C_{\bm B^\top \bm X}, \quad a.s.
\end{align}
If a matrix $\bm B$ satisfies \eqref{eq:CB}, $\bm X$ and $\bm B^\top \bm X$ carry the same information about the dependence between the components of $\bm Y$. The smallest possible copula dimension reduction subspace is called the central copula subspace.

\begin{Definition}[Central copula subspace]
Let $\mathcal{B}_{C_{\bm X}} = \{\bm B  \in \R^{p \times p}: C_{\bm X} = C_{\bm B^\top \bm X} \; a.s.\}.$ If $\mathcal{S}_{C_{\bm X}} = \bigcap_{\bm B \in \mathcal{B}_{C_{\bm X}}} \mathcal{S}(\bm B)$ is a copula dimension reduction subspace, it is called the {\bfseries central copula subspace}.
\end{Definition}

The following result states that the central subspace can be decomposed into the
marginal central  subspaces $\mathcal{S}_{Y_j \vert \bm X}$, $j = 1, \dots, q$,
and the central copula subspace $\mathcal{S}_{C_{\bm X}}$.

\begin{Theorem} \label{thm:copula_cs}
For the conditional copula $C_{\bm X}$ associated with the distribution of $\bm Y\mid\bm X$, for the central subspace it holds that
$\mathcal{S}_{\bm Y \mid \bm X} = \mathrm{span}\{\mathcal{S}_{Y_1 \mid \bm X}, \dots, \mathcal{S}_{Y_q \mid \bm X}, \mathcal{S}_{C_{\bm X}}\}.$
\end{Theorem}

\begin{proof}
Recall from the proof of \autoref{lem:drsklar1} that $F_{j, \bm X}$ and $C_{\bm X}$ can be expressed in terms of $F_{\bm Y \mid \bm X}$. Hence,
$ %\begin{align*}
  S_{Y_j \mid \bm X} \subset \mathcal{S}_{\bm Y \mid \bm X}, \quad \mbox{for } j
  = 1, \dots, q, \quad \mbox{and} \quad  S_{C_{\bm X}} \subset \mathcal{S}_{\bm Y
  \mid \bm X},
$ %\end{align*}
and, thus, $ \mathrm{span}\{\mathcal{S}_{Y_1 \mid \bm X}, \dots,
\mathcal{S}_{Y_q \mid \bm X}, \mathcal{S}_{C_{\bm X}}\} \subset
\mathcal{S}_{\bm Y \mid \bm X}$. It also implies that any matrix $\bm B \in
\R^{p \times d}$ with $\mathcal{S}(\bm B) = \mathcal{S}_{\bm Y \mid \bm X}$
possesses submatrices $\bm B_j \in \R^{p \times d_j}$, $\bm B_C \in \R^{p \times
d_C}$ satisfying $\mathcal{S}(\bm B_j) = \mathcal{S}_{Y_j \mid \bm X}$, $j = 1, \dots, q$, and $\mathcal{S}(\bm B_C) =
\mathcal{S}_{C_{\bm X}}$. By \autoref{lem:drsklar2},
it follows that $\mathrm{span}\{\mathcal{S}_{Y_1 \mid \bm X}, \dots, \mathcal{S}_{Y_q \mid \bm X},
\mathcal{S}_{C_{\bm X}}\}$ is a dimension reduction subspace for $F_{\bm Y \mid
\bm X}$. And since $\mathcal{S}_{\bm Y \mid \bm X}$ is the intersection of
all dimension reduction subspaces, it must hold $\mathcal{S}_{\bm Y \mid
\bm X}  \subset \mathrm{span}\{\mathcal{S}_{Y_1 \mid \bm X}, \dots, \mathcal{S}_{Y_q
\mid \bm X}, \mathcal{S}_{C_{\bm X}}\}$.
\end{proof}

\autoref{thm:copula_cs} can be seen as the central subspace version of Sklar's
theorem. Similar to \autoref{lem:drsklar1} and \autoref{lem:drsklar2}, the
decomposition in \autoref{thm:copula_cs} is not necessarily disjoint. But it is
certainly possible that the central marginal or copula  subspaces are of
smaller dimension than the central subspace of the joint distribution.
\begin{Example}
Consider a variant of the model from \autoref{ex:mv}:
\begin{align*}
  \bm Y =\bm g(\bm \beta^\top \bm X) + \operatorname{diag}(\sigma_1, \dots, \sigma_q) \bm R(\bm \gamma^\top \bm X)^{1/2}\bm \epsilon, \quad \bm \epsilon \sim \mathcal{N}(\bm 0, \bm I_{q \times q}),
\end{align*}
where $\bm R$ is a map from $\R$ to the space of correlation matrices. The
central subspace is again $\Scal_{\bm Y \mid \bm X} = \Scal(\bm \beta, \bm
\gamma)$ and has dimension $d = 2$. For all $j = 1, \dots, q$, the marginal central subspaces is
$\Scal_{Y_j \mid \bm X} = \Scal(\bm \beta)$ with $d_j = 1$ (the marginal variances do not change with $\bm X$). The conditional copula $C_{\bm X}$ is a Gaussian
copula with correlation matrix $\bm R(\bm \gamma^\top \bm X)$. Hence, the
central copula subspace is $\Scal_{C_{\bm X}} = \Scal(\bm \gamma)$ and $d_C = 1$. \hfill \qed
\end{Example}

The result in Theorem \ref{thm:copula_cs} is 
useful in practice because it facilitates inference. The popularity of copulas
is largely due to the fact that they allow to separate inference of the marginal
distributions and the copula. Similarly, inference of the central subspace can
be separated into inference of the central marginal subspaces and the 
central copula subspace. On finite
samples, smaller subspaces are usually easier to find. Further, when estimating
the conditional margins and copula separately, we can condition on
lower-dimensional covariates. This is particularly useful when the conditional
margins are estimated nonparametrically, because it allows for faster rates
of convergence.

\begin{Remark}
  The assumption $C_{\bm X} \equiv C$ or, equivalently, $C_{\bm X} \perp \bm X$
  is known as the \emph{simplifying assumption} in the copula literature
\citep[e.g.,][]{gijbels2015estimation, gijbels2015partial,gijbels2017nonparametric,  derumigny2017tests, portier2018weak} and commonly used in vine copula
  models \citep[e.g.,][]{haff2015nonparametric, spanhel2019simplified,
  portier2018weak, haff2010simplified, 
  nagler2016evading, stoeber2013simplified, nagler2024simplified}. It can be reformulated in terms of the central copula
  subspace: $\Scal_{C_{\bm X}} = \{\bm 0\}$.
\end{Remark}
%---------------------------------------------------------------------------_%

\subsection{Properties of the central copula subspace}
\label{subsec:ccs}

By definition of $C_{\bm X}$, the central copula subspace is the central
subspace of the conditional random vector $\bm U_{\bm X} \mid \bm X$, i.e.
$\mathcal{S}_{C_{\bm X}} = \mathcal{S}_{\bm U_{\bm X} \vert \bm
X}$. Therefore, all remarks made in \autoref{sec:background} about the existence of the central subspace apply. We introduced the central copula subspace as a
separate concept because it has distinct properties:
\begin{enumerate}
\item The marginal central subspaces only contain the zero element: one can
easily verify that the variables $U_{j, \bm X}$ are marginally independent
of $\bm X$ for $j = 1, \dots, q$.

\item The random variables $U_{1, \bm X}, \dots,  U_{q, \bm X}$ are unobserved
because the conditional margins $F_{Y_1 \mid \bm X}$, \dots,  $F_{Y_q \mid \bm
X}$ are unknown. This makes inference for the central copula subspace more
difficult.

\end{enumerate}

In view of \autoref{thm:cover}, we can characterize the central copula subspace by a collection of central mean subspaces. In \autoref{subsec:connection}, we already mentioned two examples of such families which we shall discuss in more detail.
The following result is an immediate consequence of \autoref{thm:cover}.
\begin{Theorem} \label{prop:G}
	Let
    \begin{align*}
		\mathcal{G}_1 &= \bigl\{g\colon \R^q \to \R, \; \bm z \mapsto  \ind \bigl(\bm z \le \bm u\bigr), \;\bm u \in [0,1]^q\bigr\}, \\
    \mathcal{G}_2 &=  \biggl\{g\colon \R^q \to \R, \; \bm z \mapsto  \prod_{j = 1}^q z_j^{k_j}, \; (k_1, \dots, k_q) \in \N^q  \biggr\}.
	\end{align*}
  Then,
$ %  \begin{align*}
		\mathcal{S}_{C_{\bm X}} = \mathrm{span}\bigl\{S_{\E\{g(\bm U_{\bm X}) \mid \bm X\}},  g \in \mathcal{G}_1\bigr\} = \mathrm{span} \bigl\{S_{\E\{g(\bm U_{\bm X}) \mid \bm X\}}, g \in \mathcal{G}_2 \bigr\}.
	$ %\end{align*}
\end{Theorem}

Each conditional expectation $\E\{g(\bm U_{\bm X})\mid \bm X\}$ in
\autoref{prop:G} summarizes a different aspect of $C_{\bm X}$. The families
$\mathcal{G}_1$ and $\mathcal{G}_2$ are such that the collection of summaries is
sufficient to characterize the conditional copula $C_{\bm X}$. The result is
quite intuitive: For $\mathcal{G}_1$, the function $\bm u \mapsto
\E\{\ind(U_{\bm X} \le \bm u) \mid \bm X\}$ is exactly the conditional copula
$C_{\bm X}$. The collection $\mathcal{G}_2$ relates to the fact that the
collection of all moments completely characterizes the distribution of a bounded
random vector \citep{hausdorff1921}. In our context, we do not even require all
moments, but only those where at least two exponents $k_j$ are non-zero. This
suffices because $\E(U_{j, \bm X}^{k}\mid \bm X)$ is independent of $\bm X$ for
all $k \ge 0$; see property (i) above.

The dependence summaries induced by $\mathcal{G}_1$ and $\mathcal{G}_2$ relate
to popular conditional measures of association when $q = 2$
\citep[see,][]{gijbels2011}. The conditional Spearman's $\rho$ can be expressed
as $\rho_{\bm X} = 12\E\{U_{1, \bm X} U_{2, \bm X} \vert \bm X\}- 3$ which
relates to the element of $\mathcal{G}_2$ where $k_1 = k_2 = 1$. The conditional
Blomqvist's $\beta$ can be expressed as $\beta_{\bm X} = 1 - 4\E\{\ind(U_{1, \bm
X} \le 0.5, U_{2, \bm X} \le 0.5) \vert \bm X\}$ and relates to the element of
$\mathcal{G}_1$ with $u_1 = u_2 = 0.5$. In parametric models, these
 measures often suffice to characterize the whole dependence
structure. In other cases, the dependence measure is the only quantity of
interest. We discuss sufficient dimension reduction in this simplified setting
in the following section.

 	\section{Sufficient dimension reduction for conditional dependence measures}
\label{sec:etadr}

We now state the sufficient dimension reduction problem for conditional association measures and  define central subspaces in this context. We discuss some popular measures and show that their central subspaces can be expressed in terms of central mean subspaces.

The simplest way to analyze dependence between two variables is to summarize it
in a single number, a \emph{measure of association} \citep[see, e.g.,][Chapter
5]{nelsen06}.
Any dependence measure $\eta$ satisfying R\'enyi's axiom F (invariance with respect to strictly increasing marginal transformation) \citep[see,][]{renyi1959} can be written as a
functional $T_\eta$ of the copula: $\eta = T_\eta(C)$
\citep[see,][]{schweizer1981}. The corresponding conditional measure of
dependence $\eta_{\bm X}$ is defined as the same functional $T_\eta$ applied to
the conditional copula: $\eta_{\bm X} = T_\eta(C_{\bm X})$. The dimension
reduction problem is now to find a matrix $\bm B \in R^{p \times d}$, with $d \le p$, such that
\begin{align}
	\eta_{\bm X} = \eta_{\bm B^\top \bm X} \quad a.s.
      \label{eq:etadr}
\end{align}

We define an \emph{$\eta$ dimension reduction subspace} as the column space
$\mathcal{S}(\bm B)$ of any matrix satisfying \eqref{eq:etadr}. Our goal is to
find the smallest of these spaces.

\begin{Definition}[Central $\eta$ subspace] Let $\mathcal{B}_{\eta_{\bm X}} =
\{\bm B \in \R^{p \times p} \colon \bm \eta_{\bm X} = \eta_{\bm B^\top \bm X} \;
a.s\}.$ If $\mathcal{S}_{Y \vert \bm X} = \bigcap_{\bm B \in
\mathcal{B}_{\eta_{\bm X}}} \mathcal{S}(\bm B)$ is a $\eta$
dimension reduction subspace, it is called the {\bfseries central $\eta$
subspace}. \end{Definition}\noindent

In \autoref{subsec:ccs} we saw two examples of measures whose central $\eta$
subspace is the central mean subspace related to a conditional expectation
$\E\{g(U_{1, \bm X}, U_{2, \bm X}) \mid \bm X\}$. These examples belong to a
more general class of association measures whose central subspaces can be
expressed as central mean subspaces. Consequently, we can estimate the central
$\eta$ subspace using methods that are suitable for finding the central mean
subspace.

\autoref{ex:eta} lists four popular dependence measures that are a linear
functional of the conditional copula in the sense that there exists
$g_\eta\colon [0,1]^2 \to \R$ such that $\eta_{\bm X} = \int g_\eta(u_1, u_2)
dC_{\bm X}(u_1, u_2) = \E\{g_\eta(U_{1, \bm X}, U_{2, \bm X}) \mid \bm X\}$.
Hence, $\mathcal{S}_{\eta_{\bm X}} = \mathcal{S}_{\E\{g_\eta(U_{1, \bm X}, U_{2,
\bm X} \mid \bm X)\}}$.
\begin{Example} \label{ex:eta} \quad \\[-12pt]
\begin{itemize}
\item Spearman's $\rho$: $g_\rho(u_1, u_2) = 12(u_1 -0.5)(u_2 - 0.5)$,
\item Blomqvist's $\beta$: $g_\beta = 1 - 4\ind(u_1 \le 0.5, u_2 \le 0.5)$,
\item Gini's $\gamma$: $g_\gamma = 2 \bigl(\vert u_1 + u_2 - 1 \vert - \vert u_1 - u_2 \vert\bigr)$,
\item van der Waerden's coefficient: $g_{\omega} = \Phi^{-1}(u_1) \Phi^{-1}(u_2)$, where $\Phi^{-1}$ is the standard normal quantile function. \hfill \qedsymbol
\end{itemize}
\end{Example}

All measures above belong to the family of \emph{concordance measures}
\citep{nelsen06}. The conditional Kendall's $\tau$ is another popular
concordance measure  that is not of the form $\E\{g_\eta(U_{1, \bm X}, U_{2,
\bm X}) \mid \bm X\}$. We can still relate it to a central mean subspace, but in
a more complicated fashion. The conditional Kendall's $\tau$ can be expressed as
$
\tau_{\bm X} = 4\int C_{\bm X}(u_1, u_2)d C_{\bm X}(u_1, u_2) - 1,
$
see \citet{gijbels2011}. Let  $(U_{1, \bm X}, U_{2, \bm X}) \sim C_{\bm X}$, $(\bar U_{1, \bm X}, \bar U_{2, \bm X})\sim C_{ \bm X}$ be two vectors such that
$
(U_{1, \bm X}, U_{2, \bm X}) \perp (\bar U_{1, \bm X}, \bar U_{2, \bm X}) \mid \bm X. 
$
Then,
\begin{align*}
		\int C_{\bm X}(u_1, u_2) dC_{\bm X}(u_1, u_2) = P\bigl(U_{1, \bm X} \le \bar U_{1, \bm X}, U_{2, \bm X} \le \bar U_{2, \bm X} \mid \bm X \bigr),
\end{align*}
see \citet[p.\ 160]{nelsen06} for a derivation in the unconditional case. Hence, $\tau_{\bm X} = \tau_{\bm B^\top \bm X}$ a.s., if and only if, a.s.,
    \begin{align} \label{eq:taudr}
%    \begin{aligned}
	%&
\E\bigl\{g_\tau \bigl(U_{1, \bm X}, U_{2, \bm X}, 
\bar U_{1, \bm X}, \bar U_{2, \bm X}\bigr) \mid \bm X \bigr\}   
    =\;& \E\bigl\{g_\tau \bigl(U_{1, \bm X}, U_{2, \bm X}, \bar U_{1, \bm X}, \bar U_{2, \bm X}\bigr) \mid \bm B^\top \bm X  \bigr\}, %\quad a.s.,
    %\end{aligned}
\end{align}
where
$
g_\tau(u_1, u_2, \bar u_1, \bar u_2) = 4\ind(u_1 \le \bar u_2, u_2 \le \bar u_2) - 1.
$
Although the expressions in \eqref{eq:taudr} are more difficult to estimate due to the presence of conditionally independent copies, they are still conditional expectations. Hence, the central $\tau$ subspace is also a central mean subspace:
\begin{align*}
  \mathcal{S}_{\tau_{\bm X}} = \mathcal{S}_{\E\{g_\tau(U_{1, \bm X}, U_{2, \bm X}, \bar U_{1, \bm X}, \bar U_{2, \bm X} \mid \bm X)\}}.
\end{align*}

\section{Differences and similarities between central subspaces for conditional dependence}
\label{subsec:differences}

Since the conditional measures of association $\eta_{\bm X}$ discussed in the previous section are functionals of the conditional copula, it generally holds that $\mathcal{S}_{\eta_{\bm X}} \subseteq \mathcal{S}_{C_{\bm X}}$.   This implies that the dimension can possibly be reduced further when a measure of association is of concern. Another question concerns the relationship between central $\eta$ subspaces for different choices of $\eta$. Universal statements of this kind seem out of reach. However, all subspaces do coincide in many situations of practical interest. Let us substantiate this by an example.
\begin{Example} \label{ex:1}
Consider a parametric family of copulas $\mathcal{C}^{\theta} = \{C(\cdot, \cdot; \theta)\colon \theta \in \Theta \subseteq \R\}$ and let $h\colon \Omega_{\bm X} \to \Theta$ be unknown, with $\Omega_{\bm X}$ the support of $\bm X$ ($\Omega_{\bm X}\subseteq \R^q$). Then, a general semiparametric model for the conditional copula is
\begin{align}
	\bigl(U_{1, \bm X}, U_{2, \bm X}\bigr) \vert \bm X = \bm x  \sim  C\bigl\{\cdot, \cdot; h(\bm x)\bigr\}, \quad \mbox{for all } \bm x \in \R^p.
    \label{eq:semipar}
\end{align}
For most of the popular parametric families there is a one-to-one relationship
between the dependence parameter $\theta$ and the measures of concordance
$T_\eta(C^\theta)$. Examples include the Gaussian, Student t (with fixed degrees
of freedom), Farlie-Gumbel-Morgenstern (FGM), Clayton, Gumbel, Joe, and Frank
copula families.\footnote{Although analytical formulas for the relationship
rarely exist, one can check numerically that the relation is in fact
one-to-one.} When the relationship between $\theta$ and $T_\eta(C^\theta)$ is
one-to-one, the covariate $\bm X$ always affects the copula and $\eta$ at the
same time. Similarly, for any fixed matrix $\bm B \in \R^{p \times d}$, a change
in $\eta_{\bm B^\top \bm X}$ always leads to a change in $C_{\bm B^\top \bm X}$
and the other way around. Hence,
$$\mathcal{S}_{C_{\bm X}} =
\mathcal{S}_{\rho_{\bm X}}  = \mathcal{S}_{\beta_{\bm X}} =
\mathcal{S}_{\gamma_{\bm X}} = \mathcal{S}_{\omega_{\bm X}} = \mathcal{S}_{\tau_{\bm X}}.$$
The models considered by \citet{patton2001},
\citet{acar2011}, \citet{Abegaz2012}, \citet{vatter2015}, and
\citet{fermanian2018single} are all of this type. \qed
\end{Example}\noindent On the other hand, we emphasize that this is not true in
general. Let us give an obvious counterexample where the central copula subspace
is different from central $\eta$ subspaces.

\begin{Example}
	Consider model \eqref{eq:semipar} for the Student t copula where the
association parameter is fixed, and $\theta$ refers to the degrees of freedom
parameter. The degrees of freedom has no influence on Kendall's $\tau$,
i.e., $\tau_{\bm X}$ is constant. Then
it holds trivially that $\tau_{\bm X} = \tau_{\bm B^\top \bm X}$ for $\bm B =
\bm  0$ which implies $\mathcal{S}_{\tau_{\bm X}} = \{\bm 0\}$. However, $\mathcal{S}_{C_{\bm X}}
\neq \{\bm 0\}$, because $C_{\bm X}$ does depend on the degrees of freedom
parameter and a change in $h(\bm X)$ will have an effect on $C_{\bm X}$. \qed
\end{Example}\noindent The next example demonstrates that it is also possible
that central $\eta$ subspaces are different for different association measures.

\begin{Example} \label{ex:2}
Consider two parametric copula families
$$\mathcal{C}_1^{\theta_1} = \{C_1(\cdot, \cdot; \theta_1)\colon \theta_1 \in \Theta_1 \subseteq \R\}, \quad  \mathcal{C}_2^{\theta_2} = \{C_2(\cdot, \cdot; \theta_2)\colon \theta_2 \in \Theta_2 \subseteq \R\}.$$
For $\theta_1 \in \Theta_1, \theta_2 \in \Theta_2$, and a mixing parameter $\pi \in [0, 1]$, we define the mixture of the two copulas
$ %\begin{align*}
	 C(u_1, u_2; \theta_1, \theta_2, \pi) = \pi C_1(u_1, u_2; \theta_1) + (1 - \pi) C_2(u_1, u_2; \theta_2).
$ %\end{align*}
Let $p = p_1 + p_2$, $h_1\colon \R^{p_1} \to \Theta_1$, $h_2\colon \R^{p_2} \to \Theta_2$, $\pi \colon \R^{p_2}  \to [0, 1]$. We consider the following conditional copula model: for all $\bm x =  (\bm x^{(1)}, \bm x^{(2)}) \in \R^p$,
\begin{align*}
	\bigl(U_{1, \bm X}, U_{2, \bm X}\bigr) \vert \bm X = \bm x  \sim  C\bigl\{\cdot, \cdot; h_1(\bm x^{(1)}),h_2(\bm x^{(1)}), \pi(\bm x^{(2)})\bigl\}.
\end{align*}
In this model, the first component of the covariate vector, $\bm x^{(1)}$, affects the strength of dependence in both parts of the mixture. The second component, $\bm x^{(2)}$, controls the weighting between the two parts.

We have $\mathcal{S}_{C_{\bm X}} = \mathrm{span}(\bm e^{(1)}, \bm e^{(2)})$, where $\bm e^{(2)}$ is a vector with 1 in the first $p_1$ components and all remaining components equal 0, and $\bm e^{(2)} = \bm 1 - \bm e^{(1)}$. Assume further that $h_1$ and $h_2$ are such that for all $\bm x^{(1)} \in \R^{p_1}$,
\begin{align} \label{eq:hh1}
	\int g_\rho(u_1, u_2) dC_1\{u_1, u_2; h_1(\bm x^{(1)})\}
	= \int g_\rho(u_1, u_2) dC_2\{u_1, u_2 \cdot; h_2(\bm x^{(1)})\}.
\end{align}
Since $\pi(\bm x^{(2)}) + (1 - \pi(\bm x^{(2)})) = 1$, the Spearman's $\rho$ of $C$ is not affected by $\bm x^{(2)}$. Hence, $\mathcal{S}_{\rho_{\bm X}} = \mathrm{span}(\bm e_1)$.
However, it is generally not the case that
\begin{align} \label{eq:hh2}
	\int g_\beta(u_1, u_2) dC_1\{u_1, u_2; h_1(\bm x^{(1)})\}
	= \int g_\beta(u_1, u_2) dC_2\{u_1, u_2 \cdot; h_2(\bm x^{(1)})\}
\end{align}
holds at the same time. If it does not, we have $\mathcal{S}_{\rho_{\bm X}} \neq \mathcal{S}_{\beta_{\bm X}}$. This is easy to check when analytical expressions for the relationship between the measures and the copula parameter are available. For example, it is not possible that \eqref{eq:hh1} and \eqref{eq:hh2} hold at the same time when $\mathcal{C}_1^{\theta_1}$ is the Gaussian family and $\mathcal{C}_2^{\theta_2}$ is the FGM family. \qed
\end{Example}\noindent

We conclude that the various central subspaces may differ in conditional dependence models.  But in most situations of practical interest, they  coincide.

 	\section{Estimation of central subspaces in conditional dependence models}
\label{sec:estimation}

In  Sections \ref{sec:sklar} and \ref{sec:etadr} we have seen that sufficient dimension reduction problems in conditional dependence models can be transformed to equivalent dimension reduction problems in a classical (mean) regression setting. A large number of methods for the estimation of central (mean) subspaces have been proposed \citep[see, e.g.,][]{ma2013review}.
We focus on a particular instance of nonparametric estimators: the  \emph{outer
product of gradients (OPG)} method which was introduced by \citet{xia2002} and
is based on an idea of \citet{haerdle1989}. The main motivation for this choice
is the simplicity and generality of the OPG method. It is often used as a
starting point for other techniques that need to be initialized
with a consistent estimate of the central subspace \citep[e.g.,][]{xia2002,
xia2007, luo2014}. Of course, all other methods, e.g., the ones from \citet{ma2012} or \citet{huang2017}, are similarly applicable in the
context of conditional dependence models.

\subsection{The OPG matrix and central subspaces}

Let $\Gcal \subset \{g\colon \R^{q} \to
\R\}$ be a set of functions and we want to identify the space $\Scal_\Gcal =
\mathrm{span}\{\Scal_{\E\{g(\bm Y) \mid \bm X\}}\colon g \in \Gcal\}$.
For central mean subspaces, $\Gcal$ is a singleton. A larger, but
finite number of suitably chosen functions suffice for central subspaces, see \Cref{thm:cover} and the following remarks.
In the following, we thus assume $|\Gcal| < \infty$ for simplicity \citep[this can be
relaxed by randomizing over elements in $\Gcal$, see,][]{yin2011}.
Define
\begin{gather*}
%\begin{align*}
  m_g(\bm x) = \E\{g(\bm Y) \mid \bm X = \bm x\}, \qquad m_g^{\bm B}(\bm x) =
\E\{g(\bm Y) \mid \bm B^\top \bm X = \bm B^\top \bm x\},
\\
%\end{align*}
%and
%\begin{align*}
  \bm \Delta_{\Gcal}
  = \sum_{g \in \Gcal} \E\bigl\{\nabla m_g(\bm X) \nabla^\top m_g(\bm X)  \ind_{\bm X \in \Dcal_{\bm X}}\bigr\},
%\end{align*}
\end{gather*}
where $\Dcal_{\bm X} \subseteq \R^p$ is some set.
The OPG method is based on the following observation.
\begin{Lemma} \label{lem:cs}
   Suppose $\bm B_0 \in \R^{p \times d}$ is a matrix such that $\mathrm{span}(\bm B_0) = \Scal_\Gcal$ and
  \begin{align*}
    \bm \Delta_\Gcal^{\bm B_0} = \sum_{g \in \Gcal} \E\{\nabla m_g^{\bm B_0}(\bm B_0^\top \bm X) \nabla^\top m_g^{\bm B_0}(\bm B_0^\top \bm X) \ind_{\bm X \in \Dcal_{\bm X}}\}
  \end{align*}
  is of full rank.
  Let $\bm \Delta_{\Gcal} = \bm V \bm \Lambda \bm V^\top$ be the eigen-decomposition of $\bm \Delta_{\Gcal}$. Then $\Scal_\Gcal = \mathrm{span}(\bm V \bm \Lambda^{1/2} )$.
\end{Lemma}
\begin{proof}
  Recall that since $\Delta_{\Gcal}$ is symmetric, $\bm V = (\bm v_1, \dots, \bm v_p)$ is orthogonal and $\bm \Lambda$ is diagonal containing the eigenvalues on the diagonal. Assume without loss of generality that they are ordered such that $\lambda_1 \ge \ldots \ge \lambda_d > 0$ and $\lambda_{d + 1} = \ldots = \lambda_p = 0$.
  Observe that $\mathrm{span}(\bm V \bm \Lambda^{1/2} ) = \mathrm{span}(\bm v_1, \dots, \bm v_d)$ and recall that we have the orthogonal decomposition
  $ %\begin{align*}
    \R^p = \mathrm{span}(\bm v_1, \dots, \bm v_d) + \mathrm{span}(\bm v_{d +1}, \dots, \bm v_p).
  $ %\end{align*}
  Since by assumption $ m_g(\bm X) = m_g^{\bm B_0}(\bm B_0^\top \bm X) ,$
  it holds that
  $ %\begin{align*}
    \nabla m_g(\bm X) = \nabla  m_g^{\bm B_0}(\bm B_0^\top \bm X) \bm B_0^\top,
  $ %\end{align*}
  for all $ g \in \Gcal$. We thus have
  \begin{align*}
    \bm \Delta_{\Gcal}
    = \sum_{g \in \Gcal} \bm B_0 \E\bigl\{ \nabla  m_g^{\bm B_0}(\bm B_0^\top \bm X)^\top  \nabla  m_g^{\bm B_0}(\bm B_0^\top \bm X)\bigr\} \bm B_0^\top = \bm B_0 \bm \Delta_{\Gcal}^{\bm B_0} \bm B_0^\top.
  \end{align*}
  Let $\bm v \neq \bm 0$ be some vector. Then because $\bm \Delta_{\Gcal}^{\bm B_0}$ has full rank, we have
  \begin{align*}
    \bm v \in \mathrm{span}(\bm v_{d + 1}, \dots, \bm v_p)  \quad \Leftrightarrow \quad\bm v^\top \bm \Delta_{\Gcal} \bm v = 0 \quad \Leftrightarrow \quad \bm v^\top \bm B_0 \bm \Delta_{\Gcal}^{\bm B_0} \bm B_0^\top \bm v = 0  \quad \Leftrightarrow \quad  \bm v \perp \Scal_\Gcal.
  \end{align*}
  The orthogonal complement of $\mathrm{span}(\bm v_{d + 1}, \dots, \bm v_p)$ is $\mathrm{span}(\bm v_1, \dots, \bm v_d)$, so negating the statements above gives
  \begin{align*}
    \bm v \in \mathrm{span}(\bm v_{1}, \dots, \bm v_d) \quad \Leftrightarrow \quad   \bm v \in \Scal_\Gcal. \tag*{\qedhere}
  \end{align*}
\end{proof}
The assumption that $ \bm \Delta_\Gcal^{\bm B_0}$ is of full rank can be deciphered as follows: (i) $\bm B_0$ does not contain any redundant columns, (ii) the restriction of $m_g$ to the set $\Dcal_{\bm X}$ covers all relevant directions of variation. The lemma implies that the central subspace is spanned by the eigenvectors of the OPG matrix corresponding to non-zero eigenvalues. In particular, if $\dim(\Scal_\Gcal) = d$, the first $d$ eigenvectors of $\bm \Delta_\Gcal$
span $\Scal_\Gcal$.

\subsection{An adaptive OPG method for estimating central subspaces}

We introduce an adaptive variant of the method, similar to the ones proposed by \citet{xia2007} and
\citet{yin2011}.
Because our main consistency result may be of independent interest, we present the method and its convergence properties in a general setting first.
Estimation of central dependence subspaces is discussed in the following section.

Let $\xix = \bm X_i - \bm x$ and $\bm B \in \R^{p \times
r}$, $1 \le r \le p$. To estimate the gradients $\nabla m_g(\bm x)$, define the local-linear regression estimator
\begin{align} \label{eq:ll_def}
\begin{pmatrix}
  \wh m_g(\bm x) \\ \wh \nabla m_g(\bm x)
\end{pmatrix} &= \arg\min_{\bm \beta} \sum_{i = 1}^n \biggl\{g(\bm Y_i) - \bm \beta^\top \begin{pmatrix}
  1 \\ \xix
\end{pmatrix} \biggr\}^2K_h({\bm B}^\top \xix),
\end{align}
where $K_h(\bm z) = h^{-r}\prod_{k = 1}^rK(z_k/h)$ for a univariate kernel function $K$.
By $K(\bm z)$ we denote the multivariate product kernel, i.e. $K(\bm z)=\prod_{k = 1}^r K(z_k)$.

The corresponding estimator for the OPG matrix is
\begin{align}
  \wh{\bm \Delta}_\Gcal &= \sum_{g \in \Gcal} \frac 1 n \sum_{i = 1}^n \wh \nabla m_g(\bm X_i) \wh \nabla^\top m_g(\bm X_i)  \ind(\bm X_i \in \Dcal_{\bm X}),  \label{eq:deltahat_def}
\end{align}
for some compact set $\Dcal_{\bm X} \subset \R^{p}$ introduced to stabilize the estimator.
An estimate for the central subspace $\Scal_{\Gcal}$ can be obtained by computing the eigenvectors of $\wh{\bm \Delta}_\Gcal$ corresponding to the $d$ largest eigenvalues.
However, because the gradients are estimated nonparametrically, this estimate $\Scal_{\Gcal}$ suffers from the curse of dimensionality. If $p \gg d$, the convergence will be extremely slow.

To overcome this, we consider the following adaptive procedure. Supposing for simplicity that $d =
\dim(\Scal_\Gcal)$ is known, set $\wh {\bm B}^{(0)} = \bm I_{p \times p}$, $h_0
\sim (\ln n/ n)^{1/(6 + p)}$ and fix some $\rho \in (0, 1)$ and $h_\infty \in (0, \infty)$. For $t = 0, 1,  2, \dots$:
\begin{enumerate}[label=\arabic*.]
  \item Set $h_{t} = \max\{\rho h_{t - 1}, h_\infty\}$.
  \item Compute $\wh \nabla m_g(\bm x)$ in \eqref{eq:ll_def} with $\bm B =
  \wh{\bm B}^{(t)}$ and $h = h_t$ and compute $\wh{\bm \Delta}_{\Gcal}$ as in \eqref{eq:deltahat_def}.
  \item Define $\wh{\bm B}^{(t + 1)} = \wh{\bm V} \mathrm{diag}\{s(\wh {\bm \Lambda})\}$ where  $\wh{\bm V} \wh {\bm \Lambda} \wh{\bm V}^\top $ is the eigen-decomposition of $\wh{\bm \Delta}_{\Gcal}$, and $s \colon \R^{p \times p} \to \R^{p}$  satisfies
  % \begin{align*}
   $s(\bm \Lambda)_j = 0$ if $\Lambda_{j, j} = 0$ for $j > d$.
  % \end{align*}
\end{enumerate}
The function $s$ is introduced to stabilize the convergence of the algorithm. Practical recommendations for the hyperparameters are given in \Cref{sec:sims:methods}.

Denote by $\wh{\bm B}^{(\infty)}$ the estimate obtained after convergence of this adaptive procedure.

\begin{Assumptions} \quad \\[-12pt]
\begin{enumerate}[label=A\arabic*]
	\item \label{ass:K} The kernel $K$ is a Lipschitz, symmetric probability density with support $[-1, 1]$.
	
  \item \label{ass:fB} The density $f_{\bm B_0^\top \bm X}$ has two
  bounded continuous derivatives on $\Dcal_{\bm X}$, and \\ $\inf_{\bx \in \Dcal_{\bm
  X}} f_{\bm B_0^\top \bm X}(\bm B_0^\top \bm x) > 0$.

  \item \label{ass:invertible} The matrix valued function
  \begin{align*}
    \bm M(\bm x) = \E\biggl\{\begin{pmatrix} 1 & \Xx^\top \\
      \Xx & \Xx \Xx^\top
      \end{pmatrix} \bigg| \bm B_0^\top \bm X = \bm B_0^\top \bm x \biggr\},
  \end{align*}
with $\Xx=\bm X - \bm x$,  has eigenvalues bounded away from zero, and is twice continuously differentiable, uniformly on $\Dcal_{\bm X}$.

  \item \label{ass:mg} The functions $m_g, g \in \Gcal,$ have four continuous
  derivatives on $\Dcal_{\bm X}$.

  \item \label{ass:eigvals} The eigenvalues of $\bm \Delta_\Gcal = \bV_0 \bm \Lambda_0 \bV_0^\top$ satisfy
  $\lambda_{1} > \ldots >\lambda_{d} > \lambda_{d+1} = \ldots = \lambda_{p} = 0$,
  and the first $d$ columns of $\bm B_0 = \bV_0 \mathrm{diag}\{s(\bm \Lambda_0)\}$ span $\Scal_{\Gcal}$.
\end{enumerate}
\end{Assumptions}

To avoid
ambiguity, we shall always assume that the eigenvectors of $\wh {\bm B}^{(\infty)}$ and of $\bm B_0  = \bm V_0 \mathrm{diag}\{s(\bm \Lambda_0)\}$ are oriented such that their first component is non-negative.
Let $\| \cdot \|$ denote the operator norm for matrices. The following result is proven in \autoref{proof:adaptive}.
\begin{Theorem} \label{thm:adaptive}
  Let  $h_\infty \to 0$, $nh_\infty^{d}/\ln n\to \infty$, as $n \to \infty$, and define $r_{n, d,
  \infty} = (nh_\infty^d / \ln n)^{-1/2}$. Under \ref{ass:K}--\ref{ass:eigvals},
$ %  \begin{align*}
    \| \wh{\bm B}^{(\infty)} - \bm B_0  \| = O_p(n^{-1/2}   + h_\infty^4 + r_{n, d, \infty}^2).
$ %  \end{align*}
 \end{Theorem}
\noindent The convergence rate in \autoref{thm:adaptive} does not depend on the
number of covariates $p$, but the reduced dimension $d$. If $d \le 3$, we can
choose $h_\infty$ such that $\| \wh{\bm B}^{(\infty)} - \bm B_0  \| =
O_p(n^{-1/2})$; for example $h_\infty \sim n^{-1/8}/ \ln n$.
A similar result was proven by \citet{xia2007} (who used a slightly different adaptation strategy), but only for a specific family $\Gcal$ consisting of smoothing kernels.

\subsection{Estimation of central dependence subspaces}
\label{subsec:est_eta}

We now turn to estimation of the central dependence subspaces. We consider
a finite set of functions $\Gcal$ such that
$\Scal_\Gcal = \mathrm{span}\{\Scal_{\E\{g(\bm U_{\bm X}) \mid \bm X\}}\}$.
For central $\eta$ subspaces, $\Gcal$ is a singleton. A larger, but
finite number of suitably chosen functions suffice for central copula subspaces, see \Cref{prop:G} and the following remarks. In most practically relevant cases, the central copula subspace can be identified by a single function;  see \Cref{subsec:differences}.

Assume for now that the conditional margins are known.
For $\ell = 1, \dots, q$, $i = 1, \dots, n$, set $ U_{\ell, \bm X_i} =
F_{Y_\ell \mid \bm X}(Y_{i, \ell} \mid \bm X_i)$. Define
\begin{align*}
  \begin{pmatrix}
    \wh m_g(\bm x) \\ \wh \nabla m_g(\bm x)
  \end{pmatrix} &= \arg\min_{\bm \beta} \sum_{i = 1}^n \biggl\{g({\bm U}_{\bm X_i}) - \bm \beta^\top \begin{pmatrix}
    1 \\ \xix
  \end{pmatrix} \biggr\}^2K_h({\bm B}^\top \xix), \\
  \wh{\bm \Delta}_\Gcal &= \sum_{g \in \Gcal} \frac 1 n \sum_{i = 1}^n \wh \nabla m_g(\bm X_i) \wh \nabla^\top m_g(\bm X_i)  \ind(\bm X_i \in \Dcal_{\bm X}).
\end{align*}
Setting $\bY_i = \bm U_{\bm X_i}$ in \autoref{thm:adaptive}, we get
\begin{align*}
  \| \wh{\bm B}^{(\infty)} - \bm B_0  \| = O_p(n^{-1/2}   + h_\infty^4 + r_{n, d, \infty}^2).
\end{align*}
The right-hand side is of order $O_p(n^{-1/2})$ if $d \le 3$ and $h_\infty$ appropriately chosen, see the comments after \autoref{thm:adaptive}.

The situation is more complicated when margins have to be estimated, because the estimation error propagates. When a correctly specified parametric model is estimated, we can expect these errors to contribute at most a term of order $O_p(n^{-1/2})$. In the non-parametric case, we can use the adaptive OPG method to estimate the central marginal subspaces. The convergence rate of resulting marginal estimators would be $ O_p(h_\ell^2 + (nh^{d_\ell}/\ln n)^{-1/2})$, where $h_\ell$ and $d_\ell$ are the bandwidth and dimension for estimating the $\ell$-th marginal distribution. This is of larger order than $O_p(n^{-1/2})$ no matter the choice of $h_\ell$. Some preliminary considerations suggest that the error may still be negligible when higher-order properties of the marginal estimators are taken into account. A complete analysis would involve fairly complicated third-order U-statistics. We shall not pursue this any further and refer to the simulations in the next section for an empirical evaluation.

	\section{Simulation study} \label{sec:sims}

In the following section, we present a simulation study to illustrate the performance of the proposed estimators. We compare the non-adaptive OPG estimator with a single iteration, the adaptive OPG estimator, and a parametric estimator. The simulation study is based on the following setup.

\subsection{Setup}

The goal of our study is to identify situations where sufficient dimension can be expected to work. To achieve this, we use a relatively simple simulation model for the covariates and conditional dependence and investigate the effect of its hyperparameters.

\subsubsection{Simulation model}

\begin{figure}
  \centering
  \includegraphics[width=\textwidth]{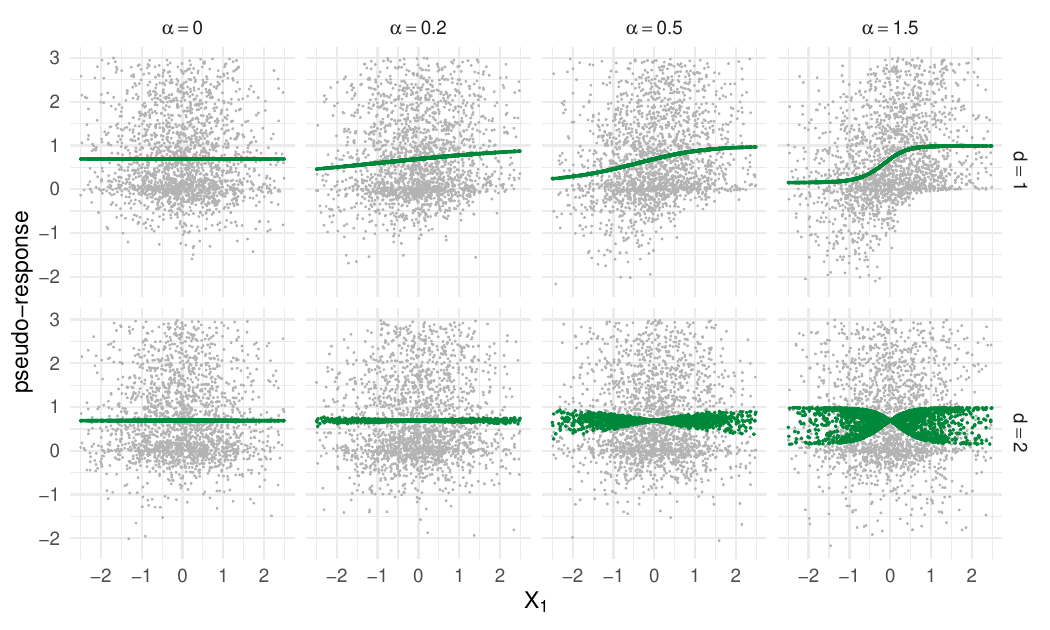}
  \caption{Scatter plot of the active covariate $X_1$ against the pseudo-response for Spearman's $\rho$ under known margins (gray) and the conditional Spearman's $\rho$ (green).}
  \label{fig:links} 
\end{figure} 

In particular, we draw $\bm X \sim \mathcal N(0, \bm \Sigma)$ with $\Sigma = \frac{1}{2} \bm I_{p \times  p} + \frac{1}{2} \bm 1_{p \times p}$ (where $\bm 1_{p \times p}$ denotes the $p\times p$ matrix with all elements equal to one) and define the conditional distribution of $(U_{1, \bm X}, U_{2, \bm X})$ by a parametric copula (Gaussian and Clayton) with conditional Kendall's $\tau$ given $\bm X = \bm x$ given as
\begin{align*}
  \tau(\bx) & = \frac{1}{2} + \frac{2}{5} \prod_{j = 1}^d \tanh(\alpha x_j).
\end{align*}
Here, $\alpha \in \R$ is a parameter that controls the strength of the signal, and $\bm e_j$ is the $j$th vector of the standard basis in $\R^p$. 
The central dependence subspaces all coincide and equal $S_{C_{\bm X}} = \mathrm{span}(\bm e_1, \dots, \bm e_d)$. 
The results were found to be relatively robust to changes in the copula family, covariate distribution, and correlation among the covariates in preliminary experiments. 
Lastly, the responses $Y_1, Y_2$ are set to
\begin{align*}
  Y_1 &=    \frac{1}{5}  X_4^2 + \frac{1}{5} X_5^2 + \Phi^{-1}(U_{1, \bm X}), \qquad Y_2 = - X_2 -  \frac{1}{5}  X_4^2 +  \Phi^{-1}(U_{2, \bm X}).
\end{align*}
The central marginal subspaces are  $\Scal_{Y_1 \mid \bm X} = \mathrm{span}(\bm e_4, \bm e_5)$ and $ \Scal_{Y_2 \mid \bm X} = \mathrm{span}(\bm e_2, \bm e_4)$. Both are different from the central copula subspace.

To get a sense of the difficulty of the problem, we plot pseudo-response for Spearman's $\rho$, i.e., $g(U_{1, \bm X}, U_{2, \bm X}) = 12 (U_{1, \bm X}- \frac 1 2)(U_{2, \bm X} - \frac{1}{2})$ as a function of $X_1$ in \Cref{fig:links} (gray points). The conditional Spearman's rho is shown as green points. The left panel ($\alpha = 0$) correspond to no effect; this is how the data looks in directions orthogonal to the central copula subspace. The other panels show increasing an effect of the covariate on the conditional dependence from left to right. 
In all upper panels, we observe that the signal-to-noise ratio is relatively low --- even for the strongest signal ($\alpha = 1.5$), where $\rho(X_1)$ changes rapidly from almost 0 to 1 in the center of the covariate distribution.
In the lower panel ($d = 2$), the conditional Spearman's $\rho$ is also modulated by $X_2$ which further obscures the relationship between $X_1$ and the pseudo-response.
The signal-to-noise ratio is further decreased when the margins are estimated.
This is in stark contrast to the well-behaved settings that methods for estimating the central mean subspace are usually evaluated in, where the signal-to-noise ratio is often very high. This is an inherent difficulty of estimating conditional dependence, because a single observation carries much less information.

\subsubsection{Estimators} \label{sec:sims:methods}

As explained in \Cref{subsec:differences}, the central copula subspace and central $\eta$ subspaces coincide in most cases of practical interest. We therefore focus on the 
performance of three methods for estimating the central $\eta$ subspace:
\begin{itemize}
  \item {\ttfamily par}: parametric estimator in a correctly specified model; numerical optimization starting at the true value of $\bm B$.
  \item {\ttfamily OPG1}: The OPG estimator for the central $\eta$ subspace after a single iteration.
  \item {\ttfamily OPGA}: The adaptive OPG estimator for the central $\eta$ subspace.
\end{itemize}

The parametric estimator is included as a (practically infeasible) baseline to assess the limits of what can be achieved in a given setting. 
The OPG estimator is implemented with the hyperparameters
\begin{align*}
  h_0 = n^{-1/(6 + p)}, \quad \rho = n^{-1/(12 + 2p)}, \quad h_\infty = n^{-1/(4 + d)},
\end{align*}
in line with the recommendations of \citet{xia2007}. Additionally, we use the stabilizing transformation
\begin{align*}
  s(\bm \Lambda) = (s_1, \dots, s_1, s_{d + 1}, \dots, s_p), \quad \text{where} \quad s_j = \Lambda_{jj}\left(\frac 1 2 + \frac{1}{2 \sum_{j = 1}^p \Lambda_{j, j}} \right),
\end{align*}
which dramatically improved the convergence behavior, especially in small-sample and weak signal settings.

In addition, we consider three settings for the conditional margins $F_{Y_j \mid \bm X}$:
\begin{itemize}
\item {\ttfamily known}: The conditional margins are known.
\item {\ttfamily parametric}: The conditional margins estimated by maximum likelihood in the correctly specified parametric model.
\item {\ttfamily nonparametric}: The conditional margins estimated by first estimating the marginal central mean subspace using OPG and then using a local linear estimator with kernel smoothing for $Y_j$, i.e.,
\begin{align*}
  \wh F_{Y_j \mid \bm X}(y \mid \bm x)  &= \argmin_{\alpha} \min_{\bm \beta} \sum_{i = 1}^n \biggl\{ \int_{\infty}^y K_b(s - Y_{i, j}) ds - \alpha - \bm \beta^\top \xix \biggr\}^2 K_h({\bm B}^\top \xix).
  \end{align*}
  The bandwidths for the OPG part are as above and the bandwidths for the conditional marginal distributions are set to $h = n^{-1/4}$ and $b = n^{-1/2}$. This strong under-smoothing is suggested by an informal analysis of the asymptotic properties under estimated margins and seen to work well in practice.
\end{itemize}

\subsubsection{Performance measure}

Performance is measured by the metric
$$d(\widehat{\bm B} , \bm B_0 ) = \| \widehat{\bm B} \widehat{\bm B}^\top - \bm B_0 \bm B_0 ^\top \|_2,$$
where $\bm B_0 = (\bm e_1, \dots, \bm e_d)$ and $\wh {\bm B}$ is an orthonormal basis of the estimated central $\eta$-subspace. The metric takes into account the non-identifiability of the basis matrices by comparing projectors onto the corresponding subspaces.

\subsubsection{Scenarios}

We consider the following scenarios:
\begin{itemize}
  \item \emph{Sample size}: $n = 150, 400, 1000, 2500$.
  \item \emph{Covariate dimension}: $p = 5, 10, 20$.
  \item \emph{Subspace dimension}: $d = 1, 2$.
  \item \emph{Signal strength}: $\alpha = 0.2, 0.5, 1.5$.
\end{itemize}
Results are based on 100 replications for each scenario.

\begin{figure}[hp]
  \centering
  \includegraphics[width=\textwidth]{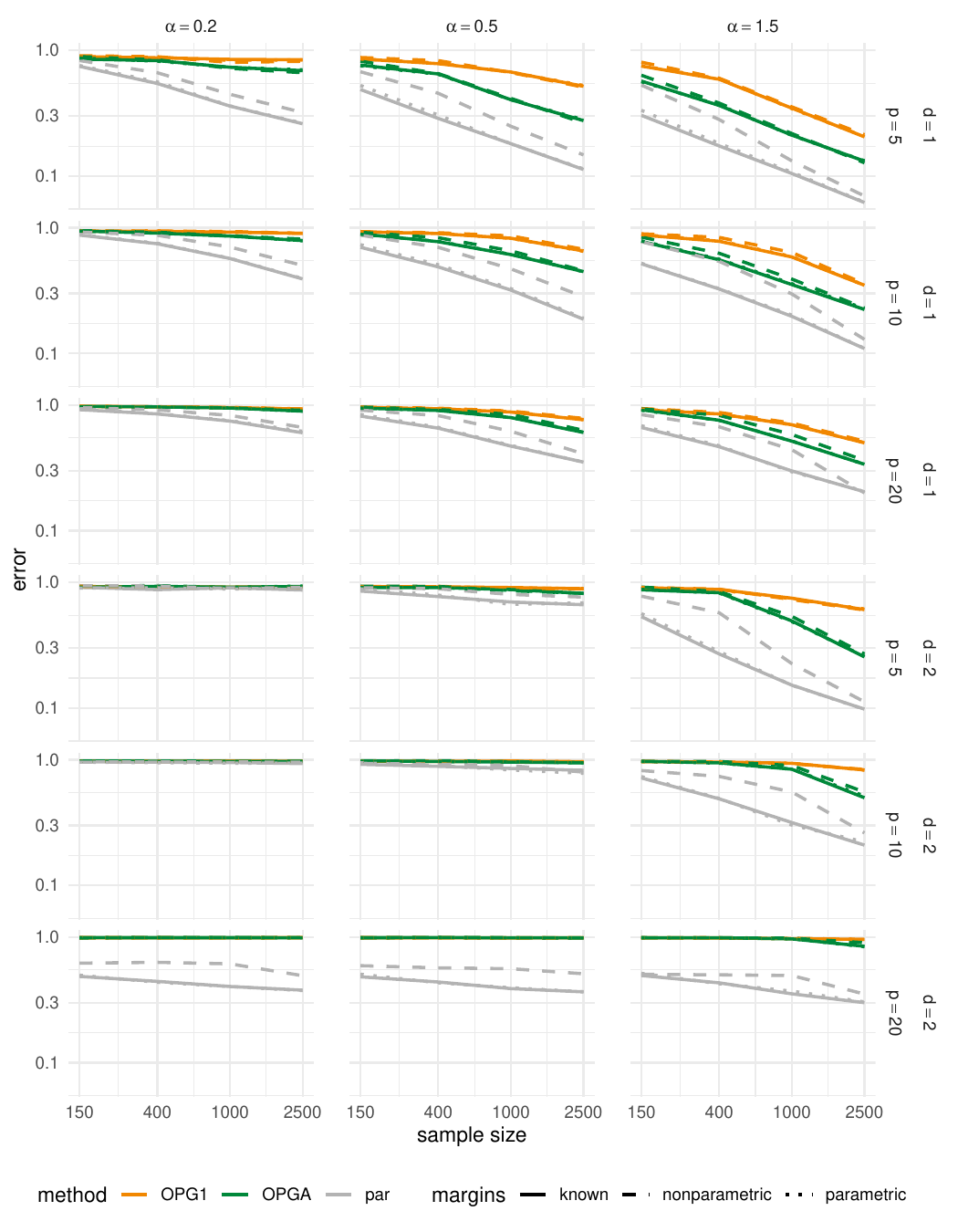}
  \caption{Average error in estimating the central $\rho$ subspace for different methods, sample sizes, signal strengths ($\alpha$), covariate dimensions ($p$), and subspace dimensions ($d$). The conditional dependence is specified as a Gaussian copula. Both axes are on log-scale such that a linear trend with slope $-1/2$ corresponds to $\sqrt{n}$-convergence.}
  \label{fig:OPG-sim}
\end{figure}

\begin{figure}[hp] 
  \centering
  \includegraphics[width=0.89\textwidth]{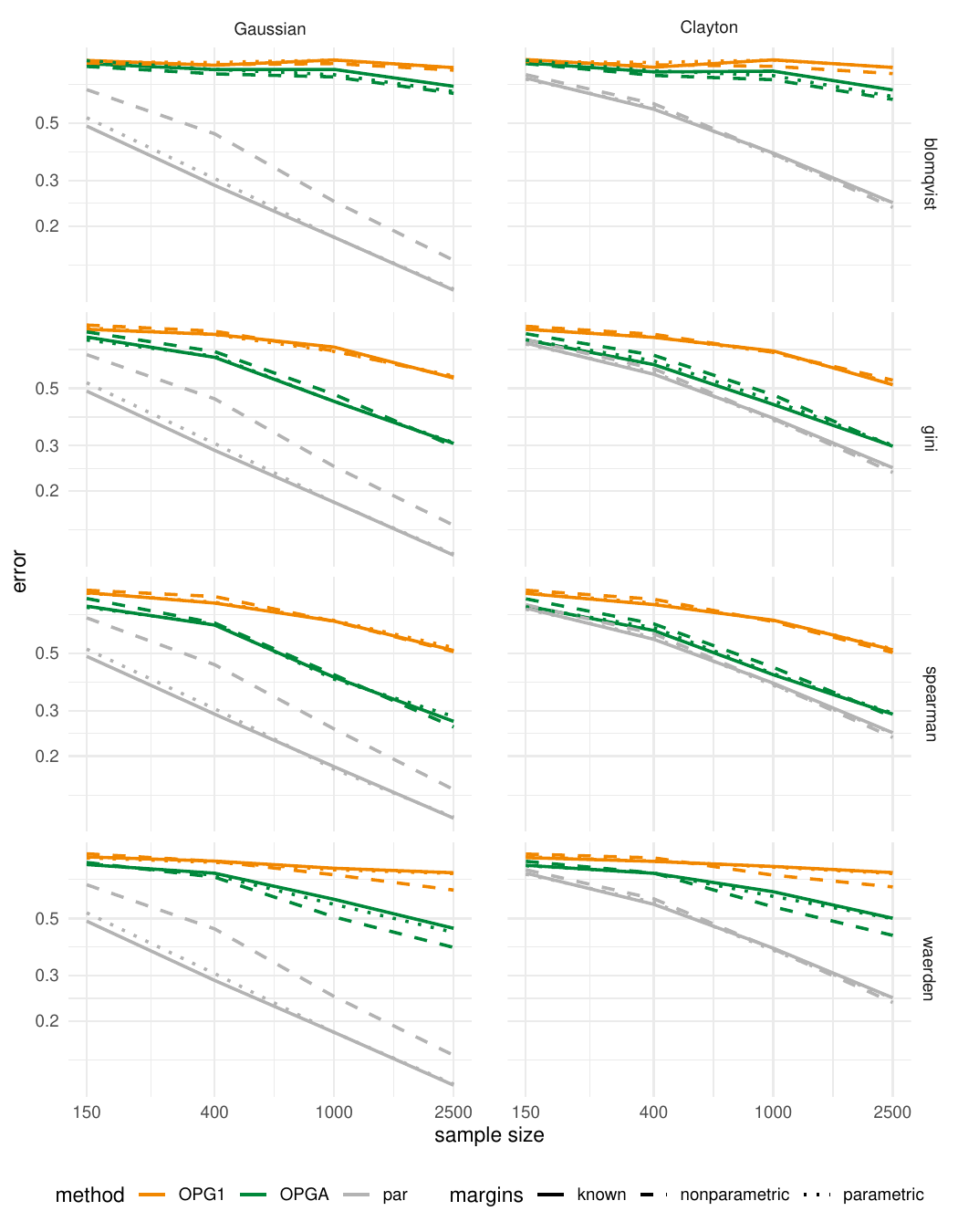}
  \caption{Average error in estimating the central subspaces for different methods, sample sizes, dependence measures, and family of the conditional copula; here, $\alpha = 0.5$, $p = 5$, and $d = 1$. Both axes are on log-scale such that a linear trend with slope $-1/2$ corresponds to $\sqrt{n}$-convergence.}
  \label{fig:OPG-sim-measure-family}
\end{figure}

\subsection{Results}

The results for estimating the central $\rho$ subspace under conditionally Gaussian dependence are shown in \Cref{fig:OPG-sim}. In each panel, the sample size is shown on the $x$-axis and the average error in estimating the central copula subspace is shown on the $y$-axis. Both are on log-scale so that a linear trend with slope $-1/2$ corresponds to $\sqrt{n}$-convergence. The different panels correspond to different signal strengths ($\alpha$), covariate dimensions ($p$), and subspace dimensions ($d$).  A few observations can be made:
\begin{itemize}
  \item As expected, the error is increasing in $d$ and $p$ and decreasing in $\alpha$.
  \item The adaptive OPG method always improves on the non-adaptive OPG method. 
  \item Some low-signal settings are so difficult that not even the parametric estimator in the correctly specified model gives reasonable results for small sample sizes or larger $d$ and  $p$. (The better performance of the parametric estimator for $p = 20$ and $d = 2$ is an artifact of early termination of the likelihood optimization which leaves the final estimates closer to the initial values).
  \item The estimators under known and parametric margins perform essentially equally well. The nonparametric margins are slightly worse, but the difference is rather small.
  \item The error of the adaptive method  appears to indeed achieve $\sqrt{n}$-convergence, at least in high-signal, large $n$ settings. The sample size where the $\sqrt{n}$-convergence starts to kick in is larger for larger $d$ and $p$.
\end{itemize}

\Cref{fig:OPG-sim-measure-family} provides additional results when varying the conditional dependence measure used for estimation and the family of the conditional copula in the case $\alpha = 0.5$, $p = 5$, and $d = 1$. Other configurations of dimension and signal strength are not shown, but lead to the same qualitative conclusions. We observe that the choice of copula family hardly affects the performance of OPG methods, but the parametric estimator is worse when the copula family is Clayton. Regarding the dependence measures, the performance is best for Spearman's $\rho$ and Gini's $\gamma$.
Estimation based on Blomqvist's $\beta$ performs much worse. This is likely due to the fact that Blomqvist's $\beta$ binarizes the data, leading to a loss in information. The van der Waerden's coefficient also performs slightly worse than the two best measures, potentially to the slightly heavier tails of the pseudo-response (sub-exponential instead of bounded).

\FloatBarrier

	\begin{appendices}
	
\newcommand{\oBo}{\bm B_0}

\section{Technical proofs}
\label{sec:Appendix}

\subsection{Preliminaries}

Let $N(\epsilon, \Fcal, L_1(P))$ be the minimal number of $L_1(P)$-balls of size
$\epsilon$ required to cover a class of functions $\Fcal$. A function $F$ is
called envelope of $\Fcal$ if $\sup_{f \in \Fcal} |f(\bm z)| \le F(\bm z)$ point-wise. For an
arbitrary probability measure $Q$ and function $F$, denote  $\|f\|_Q = \int
  |f(\bm z)| dQ(\bm z)$.
\begin{Definition}
  A class of functions $\Fcal$ is called Euclidean with respect to an envelope $F$ if there are constants $A, V < \infty$ such that $\sup_{Q} N(\epsilon\|F\|_Q, \Fcal, L_1(Q)) \le A\epsilon^{-V}$, where the supremum is over all probability measures with $\|F\|_Q > 0$.
\end{Definition}
Any class with a finite number of elements is trivially Euclidean with $V = 0$, but the classes can be much richer.
A nice property of Euclidean classes is that new classes generates from sums and products of functions from Euclidean classes are also Euclidean \citep[Lemma 2.14]{pakes1989simulation}.
The following technical lemma will be useful later on.
\begin{Lemma} \label{lem:euclidean}
  Let $r_1, \dots, r_{p} \in \N$,
  $\mathcal D$ be a bounded subset of $\R^{p}$, and $K$ satisfy Assumption \ref{ass:K}.
  For $\bs \in \mathcal D$, define
  \begin{align*}
    g_{ \bx, \bm B, h}(\bs) = \prod_{i = 1}^{p} (s_i - x_i)^{r_i}  \times \prod_{i = 1}^{p} K(\bm B_{i}^\top(\bm s - \bm x)/h),
  \end{align*}
  where $\bm B_i$ denotes the $i$-th column of $\bm B \in \R^{p \times p}$. Then the class of functions
  $ %  \begin{align*}
    \Fcal = \left\{ g_{ \bx, \bm B, h}\colon \bx \in \mathcal D, \bm B \in \R^{p \times p}, h > 0 \right\}
  $ %  \end{align*}
  is Euclidean with bounded envelope.
\end{Lemma}
\begin{proof}
  The class of functions $\{s_k \mapsto s_k - x_k \colon x_k \in \mathcal X_k\}$ is
  Euclidean by Lemmas 2.4 and 2.12 of \citet{pakes1989simulation}. Furthermore,
  the class
  $$\{\bm s \mapsto K(\bm B_i(\bs - \bx)/h) \colon \bm B_i \in \R^{p}, \bx \in  \R^{p}, h > 0 \}$$
  is Euclidean by Lemma 22 of \citet{nolan1987u}. The
  elements of the class $\Fcal$ are products of such functions and, hence, the class is also Euclidean
  by Lemma 2.14 of \citet{pakes1989simulation}. Since all terms in the product are uniformly bounded, we have a bounded envelope.
\end{proof}

Euclidean classes are small enough for uniform laws of large numbers to hold. We shall repeatedly use the following result, which also gives a rate of convergence. It is a simplifed version of Theorem II.37 of \citet{pollard2012convergence}.
\begin{Proposition} \label{prop:ep}
  For each $n \in \N$, let $\Fcal_n$ be a Euclidean class of functions with bounded envelope $F_n$ and constants not depending on $n$.
  It holds that
  \begin{align*}
    \sup_{f \in \Fcal_n} \left| \frac{1}{n} \sum_{i = 1}^n f(\bm Z_i) - \E[f(\bm Z)] \right| = O_p\left( \sqrt{\frac{\sigma_n^2 \ln n}{n}}\right),
  \end{align*}
  where $\sigma_n^2 =  \max\left\{\ln n /n, \sup_{f \in \Fcal_n} \E\{f^2(\bm Z)\}\right\}.$
\end{Proposition}
In everything that follows, we let $\bm B \in \R^{p \times p}, \bm B_0  \in \R^{p \times d}$ and define
$ %\begin{align*}
  \delta_{\bm B} = \|\bm B - (\bm B_0, \bm 0_{p \times (p - d)})  \|,
$ %\end{align*}
where $\| \cdot \|$ denotes the operator norm.
The next two lemmas are concerned with kernel smoothing with approximately low-rank bandwidth matrices. For a $d$-dimensional vector $\bm y=(y_1, \dots, y_d)$ denote $\| \bm y \|_{\infty} = \max(|y_1, \dots, |y_d|)$.

\begin{Lemma} \label{lem:kernel-lip}
  Suppose $K$ satisfies \ref{ass:K} and $\Dcal$ is a set with $A = 2\sup_{\bx \in \Dcal}\|\bx\|_2 < \infty$.
  It holds that
  \begin{align*}
    |K({\bm B}^\top\xix/h) - K({\bm B}_0^\top\xix/h)  K(0)^{p - d} | = \ind_{\|\bm B_0^\top \xix\|_\infty \le h + A \delta_{\bm B}} \times O(\delta_{\bm B}/h),
  \end{align*}
  uniformly in $\bx \in \Dcal$ and $\bm B \in \R^{p \times p}$.
\end{Lemma}
\begin{proof}
  Observe that $K(\bm B^\top \xix / h) = 0$ whenever $\|\bm B^\top \xix\|_\infty > h$.
  Since for $A = 2\sup_{\bm x \in \Dcal_{\bm X}} \|\bx\|_2$, we have
  $ %  \begin{align*}
    \|\bm B_0^\top \xix\|_\infty \le \|\bm B^\top \xix\|_\infty + A \delta_{\bm B} ,
  $ %  \end{align*}
  it holds
  $K(\bm B^\top \xix / h) = K(\bm B_0^\top \xix / h) = 0$ whenever $\|\bm B_0^\top \xix\|_\infty > h + A \delta_{\bm B}.$
  Furthermore, the Lipschitz property of the kernel implies
  \begin{align*}
    |K({\bm B}^\top\xix/h) - K({\bm B}_0^\top\xix/h)  K(0)^{p - d} | = O(\delta_{\bm B}/h). \tag*{\qedhere}
  \end{align*}

\end{proof}

\begin{Lemma} \label{lem:kernel}
  Suppose \ref{ass:K}--\ref{ass:fB} hold, and  $g$ is a bounded function such that
  \begin{align*}
    \alpha(\bm z) = \E\{g(\bm X) \mid \bm B_0^\top \bm X = \bm z\}
  \end{align*}
  is twice continuously differentiable on $\{\bm z = \bm B_0^\top \bm x\colon \bm x \in \mathcal D\}$ for some compact set $\mathcal D$.
  If $h \to 0$ and $\delta_{\bm B}/ h \to 0$, it holds uniformly in $\bm x \in \mathcal D$:
  \begin{align*}
    E[h^{p - d} g(\bm X) K_h(\bm B^\top (\bm X - \bm x))]
                                                             & = K(0)^{p - d}  \alpha(\bm B_0^\top \bm x)  f_{\oBo^\top \bm X}(\oBo^\top \bm x) + O(h^2 + \delta_{\bm B} / h), \\
    \var[h^{p - d} g(\bm X) K_h(\bm B^\top (\bm X - \bm x))] & = O(h^{-d}).
  \end{align*}
\end{Lemma}

\begin{proof}
  By \Cref{lem:kernel-lip}, we have
  \begin{align*}
    K(\bm B^\top (\bm X - \bm x)/h) & = K(\bm B_{0}^\top (\bm X - \bm x)/h) K(0)^{p - d} + \ind_{\|\bm B_0^\top \Xx\|_\infty \le h + A \delta_{\bm B}} \times O(\delta_{\bm B}/h),
  \end{align*}
  and, thus,
  \begin{align*}
     & \quad \, \E[h^{p - d} g(\bm X) K_h(\bm B^\top (\bm X - \bm x))]                                                                                                               \\
     & =  h^{-d}  \E[ g(\bm X) K(\bm B_0^\top (\bm X - \bm x) / h)]  + h^{-d}  \E[ g(\bm X) \ind_{\|\bm B_0^\top \Xx \|_\infty \le h + A \delta_{\bm B}}] \times O(\delta_{\bm B}/h) \\
     & := E_1 + E_2.
  \end{align*}
  The law of iterated expectations (first equality) and a change of variables (third equality) yield
  \begin{align*}
    E_1
     & = h^{-d} K(0)^{p - d}   \E[ \alpha(\oBo^\top \bm X)  K(\oBo^\top (\bm X - \bm x) / h)]                                                  \\
     & =   h^{-d} K(0)^{p - d}  \int \alpha(\bm s) K((\bm s - {\bm B}_0^\top \bm x) / h) f_{\oBo^\top \bm X}(\bm s) d\bm s                     \\
     & =   K(0)^{p - d}  \int_{[-1, 1]^d} \alpha(\oBo^\top  \bm x - h \bm t) K(\bm t) f_{\oBo^\top \bm X}(\oBo^\top \bm x -  h \bm t) d\bm t .
  \end{align*}
  Expanding $\alpha$ and $f_{\oBo^\top \bm X}$ around $\oBo^\top \bm x$ and noting $\int K(s) ds = 1$ and $\int sK(s) ds = 0$, we obtain
  \begin{align*}
    E_1 = K(0)^{p - d}  \alpha(\oBo^\top \bm x)  f_{\oBo^\top \bm X}(\oBo^\top \bm x) + O(h^2).
  \end{align*}
  Furthermore, since $g$ is bounded and $\delta_{\bm B} = o(h)$, we have
  \begin{align*}
    |E_2| & \le  h^{-d}  |\E[| g(\bm X) |\ind_{\|\bm B_0^\top \Xx\|_\infty \le h + A \delta_{\bm B}}] \times O(\delta_{\bm B}/h)         \\
          & \le   h^{-d} \times O(1) \times  \E[\ind_{\|\bm B_0^\top \Xx \|_\infty \le h + A \delta_{\bm B}}] \times O(\delta_{\bm B}/h) \\
          & =  h^{-d} \times O(1) \times O((h + A \delta_{\bm B})^d) \times O(\delta_{\bm B}/h)                                          \\
          & = O(\delta_{\bm B}/h).
  \end{align*}

  For the variance bound, we apply our result for expectation.
  Specifically, define $\wt g(\bm x) = g(\bm x)^2$ and $\wt K(u) = K(u)^2 / \kappa_2$, where $\kappa_2 = \int K(u)^2 du$. Observe that $\wt g$ satisfies the smoothness requirements of the lemma, and observe that $\wt K$ is also a  bounded, Lipschitz continuous, symmetric probability density with support $[-1, 1]$.
  Noting $\wt K_h(\bm x) = h^{-p} \wt K(\bm x / h)$, we have
  $$K_h(\bm x)^2 = h^{-2p} K(\bm x / h)^2 = h^{-2p} \kappa_2^{p} \wt K(\bm x / h) = h^{-p} \kappa_2^{p} \wt K_h(\bm x ) .$$
  Then
  \begin{eqnarray*}
    \lefteqn{
      \var[h^{p - d} g(\bm X) K_h(\bm B^\top (\bm X - \bm x))]
      \le  \E[h^{2(p - d)} g(\bm X)^2 K_h(\bm B^\top (\bm X - \bm x))^2]} && \\
    & = & \kappa_2^{p} \E[h^{p - 2d} \wt g(\bm X) \wt K_h(\bm B^\top (\bm X - \bm x))]
    = h^{-d}  \kappa_2^{p} \E[h^{p - d} \wt g(\bm X) \wt K_h(\bm B^\top (\bm X - \bm x))] \\
    & = & h^{-d} \times O(1)                                            = O(h^{-d}),
  \end{eqnarray*}
  by the first part of the lemma.
\end{proof}

\subsection{Uniform  approximation of the local linear estimator}

Let $\bm A^+$ denote the Moore-Penrose inverse of a matrix $\bm A$ and define $r_{n, d} = (nh^d / \ln n)^{-1/2}$.
In this section, we prove the following result.

\begin{Theorem} \label{thm:ll_refined}
  Suppose that $d < p$, $\Gcal$ and $\Hcal = \bigl\{\bm z \mapsto \E\{g(\bm Y) \mid \bm Z = \bm z \}\colon g \in \Gcal \bigr\}$
  are Euclidean classes and \ref{ass:K}--\ref{ass:mg} hold.
  Then uniformly in $g \in \Gcal, \bx \in \Dcal_{\bm X}$, and all $ \bm B \in \mathcal \R^{p \times p}$ with $\delta_{\bm B} / h \to 0$, it holds
  \begin{align*}
    \biggl\vert
    \begin{pmatrix}
      \wh m_g(\bm x) \\ \wh \nabla m_g(\bm x)
    \end{pmatrix} -
    \begin{pmatrix}
      m_g(\bm x) \\  \nabla m_g(\bm x)
    \end{pmatrix}  - {\bm v}_{n, g}^*(\bm x) \biggr\vert
     & = \begin{pmatrix}
           O(h^2 ) \\
           O(h^4)
         \end{pmatrix}
    + O_p( h^2r_{n, d} + h\delta_{\bm B} + r_{n, d} \delta_{\bm B} / h),
  \end{align*}
  where
  \begin{align*}
     & \quad \, {\bm v}_{n, g}^*(\bm x) \\
     & =
    \frac 1 {n f_{\bm B_0^\top \bm X}(\bm B_0^\top \bm x)}\sum_{i = 1}^n  \begin{pmatrix}  \epsilon_{g}(\bm Y_i, \bm X_i)\bigl[1 + \{\bm \mu(\bm x) - \bm x\}^\top \bm \Gamma^{+}(\bm x)\{\bm \mu(\bm x) - \bm X_i\}\bigr]K_h({\bm B_0}^\top\xix) \\
                                                                            \bm \Gamma^{+}(\bm x) \epsilon_{g}(\bm Y_i, \bm X_i)\{\bm X_i - \bm \mu(\bm x)\}K_h({\bm B_0}^\top\xix)
                                                                          \end{pmatrix},
  \end{align*}
  \begin{align*}
    \epsilon_{g}(\bm Y_i, \bm X_i) & = g(\bm Y_i) - \E\{g(\bm Y_i) \mid \bm X_i\},                                                               \\
    \bm \mu(\bm x)                 & = \E\{\bm X \mid \bm B_0^\top \bm X = \bm B_0^\top \bm x\}, \quad                                           \\
    \bm \Gamma(\bm x)              & = \E\{\bm X \bm X^\top \mid \bm B_0^\top \bm X = \bm B_0^\top \bm x\} - \bm \mu(\bm x) \bm \mu(\bm x)^\top.
  \end{align*}
\end{Theorem}
The result follows from a sequence of lemmas stated and proven in the following.
Equating the derivative of the local-linear estimation criterion in \eqref{eq:ll_def} to zero yields
\begin{align} \label{eq:ll_sol}
  \begin{pmatrix}
    \wh m_g(\bm x)- m_g(\bm x) \\ \wh \nabla m_g(\bm x) - \nabla m_g(\bm x)
  \end{pmatrix} = \bm S_{n, g}^{+}(\bm x) \bm \tau_{n, g}(\bm x),
\end{align}
where
\begin{align}
  \bm S_{n, g}(\bm x)    & = \frac {h^{p - d}} n \sum_{i = 1}^n \begin{pmatrix} 1 & \xix^\top \\ \xix & \xix \xix^\top \end{pmatrix} K_h(\bm B^\top\xix), \notag \\
  \bm \tau_{n, g}(\bm x) & = \frac {h^{p - d}} n\sum_{i = 1}^n\biggl\{g(\bm Y_i) - \bm \beta(\bm x)^\top
  \begin{pmatrix} 1 \\ \xix \end{pmatrix} \biggr\}
  \begin{pmatrix} 1 \\ \xix \end{pmatrix}
  K_h({\bm B}^\top\xix).
  \label{eq:taung}
\end{align}

\begin{Lemma} \label{lem:S}
  If $h \to 0$, $\delta_{\bm B}/h \to 0$, and \ref{ass:K}--\ref{ass:fB} hold, then
  \begin{align*}
     & \sup_{\bm x \in \Dcal_{\bm X}, \bm B \in \R^{p \times p}} \| \bm S_{n, g}(\bm x) - \bm S_g(\bm x) \| = O_p(h^2 + \delta_{\bm B}/h+ r_{n, d}),
  \end{align*}
  where
  \begin{align*}
    \bm S_g(\bm x) & = K(0)^{p - d} f_{\bm B_0^\top \bm X}(\bm B_0^\top\bm x) \E\biggl\{
    \begin{pmatrix}
      1   & \Xx^\top     \\
      \Xx & \Xx \Xx^\top
    \end{pmatrix}
    \bigg| \bm B_0^\top \bm X = \bm B_0^\top \bm x \biggr\}.
  \end{align*}
\end{Lemma}

\begin{proof}
  \autoref{lem:kernel} implies that
  $ %  \begin{align*}
    \E\{\bm S_{n, g}(\bm x)\} = \bm S_g(\bm x) + O(h^2 + \delta_{\bm B} / h),
  $ %  \end{align*}
  uniformly.
  By \autoref{lem:euclidean}, the  classes of functions
  \begin{align*}
    \mathcal R_{j, k} = \biggl\{\bm s \mapsto h^p  \begin{pmatrix}
                                                     1               & (\bm s - \bm x)^\top                 \\
                                                     (\bm s - \bm x) & (\bm s - \bm x) (\bm s - \bm x)^\top
                                                   \end{pmatrix}_{j, k} K_h(\bm B^\top(\bm s - \bm x)) \colon \bm x \in  \Dcal_{\bm X}, \bm B
    \in \R^{p \times p} \biggr\}
  \end{align*}
  are Euclidean for any $j, k \in \{1, \dots, p + 1\}$.
  By \autoref{lem:kernel}, we also have
  \begin{align*}
    \max_{1 \le j, k \le p + 1}\sup_{r \in \mathcal R_{j, k}} \E\{r(\bm X)^2\} = O(h^{d}),
  \end{align*}
  so that \Cref{prop:ep} yields
  \begin{align*}
    \sup_{\bm x \in \Dcal_{\bm X}, \bm B \in \R^{p \times p}} \| \bm S_{n, g}(\bm x) - \E\{\bm S_{n, g}(\bm x)\} \| = O_p(r_{n, d}),
  \end{align*}
  completing the proof.
\end{proof}

\begin{Corollary}
  Under \ref{ass:K}--\ref{ass:invertible},
  $ %  \begin{align*}
    \bm S_{n, g}^{+}(\bm x) = \bm S_g^{+}(\bm x)\{1 + o_p(1)\},
  $ %  \end{align*}
  uniformly on $\Dcal_{\bm X}$.
\end{Corollary}

\begin{Lemma}
  It holds
  \begin{align*}
    \bm S_g^{+}(\bm x) = \frac 1 {f_{\bm B_0^\top \bm X}(\bm B_0^\top\bm x)} \begin{pmatrix}
                                                                               1 + \{\bm \mu(\bm x) - \bm x\}^\top \bm \Gamma^{+}(\bm x)\{\bm \mu(\bm x) - \bm x\} & - \{\bm \mu(\bm x) - \bm x\}^\top \bm \Gamma^{+}(\bm x) \\
                                                                               -\bm \Gamma^{+}(\bm x)\{\bm \mu(\bm x) - \bm x\}                                    & \bm \Gamma^{+}(\bm x)
                                                                             \end{pmatrix},
  \end{align*}
\end{Lemma}
\begin{proof}
  Recall block inversion formula
  \begin{align*}
    \begin{pmatrix}
      \bm A & \bm B \\
      \bm C & \bm D
    \end{pmatrix}^{-1} = \begin{pmatrix}
                           \bm A^{-1} + \bm A^{-1}\bm B(\bm D - \bm C\bm A^{-1}\bm B)^{-1}\bm C\bm A^{-1} & -\bm A^{-1}\bm B(\bm D - \bm C \bm A^{-1} \bm B)^{-1} \\
                           -(\bm D - \bm C \bm A^{-1} \bm B)^{-1}\bm C\bm A^{-1}                          & (\bm D - \bm C \bm A^{-1} \bm B)^{-1}
                         \end{pmatrix},
  \end{align*}
  which we can simplify to
  \begin{align*}
    \begin{pmatrix}
      1     & \bm C^\top \\
      \bm C & \bm D
    \end{pmatrix}^{-1} = \begin{pmatrix}
                           1 + \bm C^\top(\bm D - \bm C\bm C^\top)^{-1}\bm C & -\bm C^\top(\bm D - \bm C \bm C^\top)^{-1} \\
                           -(\bm D - \bm C \bm C^\top)^{-1}\bm C             & (\bm D - \bm C \bm C^\top)^{-1}
                         \end{pmatrix}.
  \end{align*}
  Then the result follows from setting
  \begin{align*}
    \bm C = \E( \Xx\mid \bm B_0^\top \bm X = \bm B_0^\top \bm x) = \bm \mu(\bm x) - \bm x, \quad
    \bm D  = \E( \Xx \Xx^\top \mid \bm B_0^\top \bm X= \bm B_0^\top \bm x),
  \end{align*}
  which gives
  \begin{align*}
    \bm D - \bm C \bm C^\top & = \E( \Xx \Xx^\top \mid \bm B_0^\top \bm X= \bm B_0^\top \bm x)                                                                                  \\
                             & \phantom{=}- \E(\Xx \mid \bm B_0^\top \bm X= \bm B_0^\top \bm x)\E( \Xx \mid \bm B_0^\top \bm X= \bm B_0^\top \bm x)^\top                        \\
                             & = \E[\bm X \bm X^\top \mid  B_0^\top \bm X= \bm B_0^\top \bm x] - \bm \mu(\bm x) \bm x^\top - \bm x \bm \mu(\bm x)^\top - \bm x \bm x^\top \\
                             & \quad - (\bm \mu(\bm x) - \bm x)(\bm \mu(\bm x) - \bm x)^\top.
    \\
                             & = \E[\bm X\bm X^\top \mid  B_0^\top \bm X= \bm B_0^\top \bm x] - \bm \mu(\bm x) \mu(\bm x)^\top
    = \bm \Gamma(\bm x). \tag*{\qedhere}
  \end{align*}
\end{proof}

The structure of $\bm S_g^{+}(\bm x)$ will cancel the leading bias term for the gradient estimate. To see this, we first derive an approximation of $\bm \tau_{n, g}(\bm x)$ defined in \eqref{eq:taung}.  For an arbitrary function $s$, denote  $\partial_{i_1, \cdots, i_m} s(\bm u) = \partial^{m}s(\bm u)/(\partial u_{i_1} \cdots \partial u_{i_m})$.

\begin{Lemma} \label{lem:tau}
  Under the conditions of \Cref{thm:ll_refined}, it holds
  \begin{align*}
    \sup_{\bm x \in \Dcal_{\bm X}, g \in \Gcal} \left\|
    \bm \tau_{n, g}(\bm x) - h^2 \bm \rho(\bm x) - \bm v_{n, g}(\bm x) \right\|
    = O_p(h^4 + h^2r_{n, d} + h\delta_{\bm B} + r_{n, d} \delta_{\bm B} / h),
  \end{align*}
uniformly in $g \in \Gcal, \bx \in \Dcal_{\bm X}$, and  $\bm B \in \mathcal \R^{p \times p}$ with $\delta_{\bm B} / h \to 0$,
  where
  \begin{align*}
    \bm \rho(\bx)       & =  \frac{ K(0)^{p - d}}{2} \begin{pmatrix} 1 \\ \bm \mu(\bx) - \bm x \end{pmatrix} f_{\bm B_0^\top \bm X}(\oBo^\top \bx)\int_{[-1, 1]^d}  \bt^\top \nabla m_g^{\oBo}( \oBo^\top \bx) \bt  K(\bt) d\bm t \\
    \bm v_{n, g}(\bm x) & = \frac 1 n\sum_{i = 1}^n \bigl[g(\bm Y_i) - \E\{g(\bm Y_i) \mid \bm X_i\}\bigr]\begin{pmatrix} 1 \\ \xix \end{pmatrix} K_h({\bm B_0}^\top\xix) K(0)^{p -d}.
  \end{align*}
\end{Lemma}

\begin{proof}

  Decompose
  \begin{align*}
    \bm \tau_{n, g}(\bm x) & = \frac {h^{p - d}} n\sum_{i = 1}^n\biggl\{g(\bm Y_i) - \bm \beta(\bm x)^\top \begin{pmatrix}
                                                                                                             1 \\ \xix
                                                                                                           \end{pmatrix} \biggr\}\begin{pmatrix} 1 \\ \xix \end{pmatrix} K_h({\bm B}^\top\xix)                              \\
                           & =  \frac {h^{p - d}} n\sum_{i = 1}^n\bigl[g(\bm Y_i) - \E\{g(\bm Y_i) \mid \bm X_i\} \bigr]\begin{pmatrix} 1 \\ \xix \end{pmatrix} K_h({\bm B}^\top\xix)                                       \\
                           & \phantom{=}+ \frac {h^{p - d}} n\sum_{i = 1}^n\bigg[\E\{g(\bm Y_i) \mid \bm X_i\}  -  \bm \beta(\bm x)^\top \begin{pmatrix}
                                                                                                                                             1 \\ \xix
                                                                                                                                           \end{pmatrix} \biggr]\begin{pmatrix} 1 \\ \xix \end{pmatrix} K_h({\bm B}^\top\xix) \\
                           & = \bm v_{n, g}(\bm x, \bm B) + \bm a_{n, g}(\bm x, \bm B).
  \end{align*}
  We deal with the two terms separately.

  \subsubsection*{Term $\bm v_{n, g}(\bm x, \bm B)$}
  It holds
  \begin{align*}
     & \quad   \bm v_{n, g}(\bm x, \bm B)                                                                                                                                                                                                                      \\
     & = \bm v_{n, g}(\bm x)                                                                                                                                                                                                                                   \\
     & \quad + \frac {h^{- d}} n\sum_{i = 1}^n\bigl[g(\bm Y_i) - \E\{g(\bm Y_i) \mid \bm X_i\} \bigr]\begin{pmatrix} 1 \\ \xix \end{pmatrix}\left\{K\left(\frac{{\bm B}^\top\xix}{h}\right) - K\left(\frac{{\bm B}_0^\top\xix}{h}\right)  K(0)^{p -d} \right\} \\
     & = \bm v_{n, g}(\bm x) + \overline{\bm v}_{n, g}(\bm x, \bm B).
  \end{align*}
  The classes $\Gcal$ and $\Hcal$ are Euclidean by assumption. By \autoref{lem:euclidean}, also the classes
  \begin{align*}
    \mathcal V_j = \left\{\bm s \mapsto h^{p}\begin{pmatrix} 1 \\ \bm s - \bx \end{pmatrix}_j\{K_h({\bm B}^\top(\bm s - \bm x)) - K_h({\bm B_0}^\top(\bm s - \bm x)) \} \colon \bm B \in  \R^{p \times p}\colon \bm x \in \Dcal_{\bm X} \right\}
  \end{align*}
  are Euclidean. Since products of Euclidean classes are Euclidean, each coordinate of $h^d\overline{\bm v}_{n, g}(\bm x, \bm B)$ is a sample average
  indexed by a Euclidean class. By iterated expectations, each element in the
  class has zero expectation.
  By \Cref{lem:kernel-lip}, each element is bounded by
  $ %  \begin{align*}
    A \ind_{\|{\bm B}_0^\top(\bm s - \bm x)\| \le h + A \delta_{\bm B}} \times O(\delta_{\bm B} /h),
  $ %  \end{align*}
  where $A = 2 \sup_{\bm x \in \Dcal_{\bm X}} \|\bm x\|_2$.
  Since $f_{\bm B_0^\top \bm X}$ is bounded, we have
  $ %  \begin{align*}
    \E\left[ \ind_{\|{\bm B}_0^\top\xix\| \le h + A \delta_{\bm B}}^2\right] = O(h^d),
  $ %  \end{align*}
  it follows that
  \begin{align*}
    \max_{1 \le j \le p} \sup_{v \in \mathcal V_j} \E[v(\bX)^2 / \delta_{\bm B}^2]  = O_p(h^{d - 2} ).
  \end{align*}
  Again using \Cref{prop:ep}, we obtain
  \begin{align*}
    \sup_{\bm x \in \Dcal_{\bm X}, \bm B \in \R^{p \times p}, g \in \Gcal} \| \overline{\bm v}_{n, g}(\bm x, \bm B) \|/\delta_{\bm B} = O_p(r_{n, d}  / h).
  \end{align*}

  \subsubsection*{Term $\bm a_{n, g}(\bm x, \bm B)$}

  Define
  \begin{align*}
    \gamma(\bX_i, \bm x) & = \E\{g(\bm Y_i) \mid \bm X_i\}  -  \bm \beta(\bm x)^\top \begin{pmatrix}
                                                                                       1 \\ \xix
                                                                                     \end{pmatrix},
  \end{align*}
  and observe that
  \begin{align*}
    \gamma(\bX_i, \bm x)
     & = m_{g}^{\oBo}( \oBo^\top \bm X_i) - m_{g}^{\oBo}( \oBo^\top \bm x) - \nabla m_{g}^{\oBo} ( \oBo^\top \bm x) \oBo^\top \xix.
  \end{align*}
  In fact, $\gamma(\bX_i, \bm x)$ can be written as a function of $\oBo^\top \bm X_i$ and $\oBo^\top \bm x$ only, and we write $ \gamma(\bX_i, \bm x) =  \gamma_{\bm x}(\oBo^\top \bm X_i)$ in what follows.
  With this notation, we have
  \begin{align*}
    \bm a_{n, g}(\bm x, \bm B) = h^2 \frac {h^{p - d}} n\sum_{i = 1}^n \frac{\gamma_{\bm x}(\oBo^\top \bm X_i)}{h^2}\begin{pmatrix} 1 \\ \xix \end{pmatrix} K_h({\bm B}^\top\xix).
  \end{align*}
  We compute the expectation and variance of the summands similarly to the proof of \autoref{lem:kernel}.
  We have
  \begin{eqnarray*}
    \lefteqn{\E\left[h^{p - d}\frac{\gamma_{\bm x}(\oBo^\top \bX_i)}{h^2}\begin{pmatrix} 1 \\ \xix \end{pmatrix} K_h({\bm B}^\top\xix)\right]                                            } \\
    & = & h^{- d} \E\left[\frac{\gamma_{\bm x}(\oBo^\top \bX_i)}{h^2}\begin{pmatrix} 1 \\ \xix \end{pmatrix} K({\bm B}_0^\top\xix/h) K(0)^{p - d}\right]                                     \\
    && \quad + h^{- d} \E\left[\frac{\gamma_{\bm x}(\oBo^\top \bX_i)}{h^2}\begin{pmatrix} 1 \\ \xix \end{pmatrix} \{K({\bm B}^\top\xix/h) - K({\bm B}_0^\top\xix/h)  K(0)^{p - d} \}\right] \\
    & =& E_1 + E_2.
  \end{eqnarray*}
  For $E_1$, the law of iterated expectations implies
  \begin{align*}
    E_1 = K(0)^{p - d} \E\left[\frac{\gamma_{\bm x}(\oBo^\top \bX_i)}{h^2}\begin{pmatrix} 1 \\ \bm \mu(\bm X_i) - \bm x \end{pmatrix} K_h({\bm B}_0^\top\xix)\right].
  \end{align*}
  Define $\wt{\bm \mu}(\oBo^\top \bx) = \E[\bm X \mid \oBo^\top \bm X = \oBo^\top \bx] = \bm \mu(\bx)$.
  The change of variables $\bt = (\bm s - \oBo^\top \bm x)/h$ gives
  \begin{align*}
     & \quad \, E_1 / K(0)^{p - d}                                                                                                                                                                                                              \\
     & =  \int h^{-2} \gamma_{\bx}(\bs) \begin{pmatrix} 1 \\ \wt{\bm \mu}(\bm s) - \bm x \end{pmatrix} K_h(\bm s  - \oBo^\top \bm x) f_{\bm B_0^\top \bm X}(\bm s) d\bm s \\
     & = \int_{[-1, 1]^d}   h^{-2}\gamma_{\bx}( \oBo^\top \bx - h\bt)  \begin{pmatrix} 1 \\  \wt{\bm \mu}(\oBo^\top \bx - h\bt) - \bm x \end{pmatrix} K(\bt) f_{\bm B_0^\top \bm X}(\oBo^\top \bx - h\bt) d\bm t.
  \end{align*}
  Observe that
  \begin{align*}
    \gamma_{\bx}(\oBo^\top \bx) = 0, \quad \nabla \gamma_{\bx}(\oBo^\top \bx) = \bm 0, \quad  \nabla^2  \gamma_{\bx}(\oBo^\top \bz) = \nabla^2 m_g^{\oBo}(\oBo^\top \bz),
  \end{align*}
  so that a fourth order Taylor expansion of $\gamma_{\bx}$ around $\oBo^\top \bx$ yields
  \begin{align*}
    h^{-2}\gamma_{\bx}( \oBo^\top \bx - h\bt) 
    &= \frac{1}{2} \bt^\top \nabla^2 m_g^{\oBo}(\oBo^\top \bz) \bt + \frac{h}{6} \sum_{i = 1}^d \sum_{j= 1}^d \sum_{k = 1}^d t_i t_j t_k \partial_{i, j, k} m_g^{\oBo}(\oBo^\top \bx)  + O(h^2 \| \bt\|^4),
  \end{align*}
  where $\partial_{i, j, k} m_g^{\oBo}(\bs) = \partial^3 m_g^{\oBo}(\bs)/(\partial s_i \partial s_j \partial s_k)$.
  Also expanding $\wt{\bm \mu}$, and $f_{\bm B_0^\top \bm X}$ around $\oBo^\top \bx$ (up to second order) and noting $\int K(s) ds = 1$ and $\int sK(s) ds = 0$, we obtain
  \begin{align*}
    E_1 & = \frac{K(0)^{p - d}}{2} \begin{pmatrix} 1 \\ \wt{\bm \mu}(\oBo^\top \bx) - \bm x \end{pmatrix} f_{\bm B_0^\top \bm X}(\oBo^\top \bx)\int_{[-1, 1]^d}  \bt^\top \nabla^2 m_g^{\oBo}( \oBo^\top \bx) \bt  K(\bt) d\bm t + O(h^2) \\
        & =   \bm \rho(\bm x) + O(h^2).
  \end{align*}

  Next, recall
  \begin{align*}
    E_2 & = h^{-d}\E\left[\frac{\gamma_{\bm x}(\oBo^\top \bX_i)}{h^2}\begin{pmatrix} 1 \\ \xix \end{pmatrix} \{K({\bm B}^\top\xix/h) - K({\bm B}_0^\top\xix/h)  K(0)^{p - d} \}\right].
  \end{align*}
  By \Cref{lem:kernel-lip}, we have
  \begin{align*}
    |E_2| & \le h^{-d}(1 + 2A)\E\left[\left|\frac{\gamma_{\bm x}(\oBo^\top \bX_i)}{h^2} \right| \ind_{\|\bm B_0^\top \xix\|_\infty \le h + A \delta_{\bm B}} \right] \times O(\delta_{\bm B}/h).
  \end{align*}
  Further, $\|\bm B_0^\top \xix\|_\infty \le h + A \delta_{\bm B}$ and $\delta_{\bm B} = o(h)$ imply
  % \begin{align*}
  $| \gamma_{\bm x}(\oBo^\top \bX_i)| / h = O(1).$
  % \end{align*}
  Thus,
  \begin{align*}
    |E_2| & \le h^{-d}(1 + 2A) \times O(1) \times \E\left[ \ind_{\|\bm B_0^\top \xix\|_\infty \le h + A \delta_{\bm B}} \right] \times O(\delta_{\bm B}/h)
    = O(\delta_{\bm B}/h).
  \end{align*}
  Putting everything together, we have shown
  \begin{align*}
    \E[\bm a_{n, g}(\bm x, \bm B)] = h^2(E_1 + E_2) = O(h^2 + \delta_{\bm B}h).
  \end{align*}

  For the variance, we have
  \begin{eqnarray*}
    \lefteqn{\left\| \var\left[h^{p - d}\frac{\gamma_{\bm x}(\oBo^\top \bm X_i)}{h^2} \begin{pmatrix} 1 \\ \xix \end{pmatrix} K_h({\bm B}^\top\xix)\right] \right\|} \\
    & \le & h^{-2d}2(1 + A^2)\E\left[\left|\frac{\gamma_{\bm x}(\oBo^\top \bm X_i)}{h^2}\right|^2 K({\bm B}^\top\xix / h)^2 \right] \\
    & \le & O(h^{-2d})  \times \E\left[\ind_{\|\bm B_0^\top \xix\|_\infty \le h + A \delta_{\bm B}} \right]
    = O(h^{-d}),
  \end{eqnarray*}
  by the same arguments as in \Cref{lem:kernel}.
  Using \Cref{prop:ep}, we obtain
  \begin{align*}
    \left\| \bm a_{n, g}(\bm x, \bm B) - \E[ \bm a_{n, g}(\bm x, \bm B)] \right \| = O_p(r_{n, d} h^2 + \delta_{\bm B}h),
  \end{align*}
  uniformly for  $g \in \Gcal, \bm x \in \Dcal_{\bm X}$, $\bm B \in \R^{p \times p}$ with $\delta_{\bm B} / h \to 0$. This completes the proof.
\end{proof}

Now we see that all components of $\bm S_g^{+}(\bm x)  \bm \rho(\bm x)$ except the first cancel out.
This completes the proof of \Cref{thm:ll_refined}.

\subsection{Proof of \autoref{thm:adaptive}} \label{proof:adaptive}

For $t = 0$, arguments similar to \Cref{lem:S} and \Cref{lem:tau} with $p = d$ and $\bm B_0 = \wh{\bm B}^{(0)}$ yield
$ %\begin{align*}
  \| \wh{\bm \Delta}_{\Gcal} -  {\bm \Delta}_{\Gcal} \| = O_p(h^2 + r_{n, p}/h).
$ %\end{align*}
By Assumption \ref{ass:eigvals} and \citet[Theorem 2]{yu2015useful},
\begin{align*}
  \delta_{\wh{\bm B}^{(0)}}  = \| \wh{\bm B}^{(0)} - \bm B_0 \| \lesssim \| \wh{\bm \Delta}_{\Gcal} -  {\bm \Delta}_{\Gcal} \| = O_p(h_0^2 + r_{n, p}/h_0) = O_p\{(\ln n / n)^{2/(6 + p)}\}.
\end{align*}
Now consider $t > 0$ and define $r_{n, d, t} = (nh_t^d / \ln n)^{-1/2}$,  $\sigma_{n, t} = h_t^4 + h_t^2r_{n, d, t} + h_t\delta_{\wh{\bm B}^{(t)}} + r_{n, d, t} \delta_{\wh{\bm B}^{(t)}} / h_t$.
First expand
\begin{align*}
  \wh{\bm \Delta}_{g} & = \frac 1 n \sum_{i = 1}^n  \wh \nabla m_g(\bm X_i)  \wh \nabla m_g(\bm X_i)^\top                                                                                                                   \\
                      & =  \frac 1 n \sum_{i = 1}^n \nabla  m_g(\bm X_i) \nabla  m_g(\bm X_i)^\top                                                                                                                          \\
                      & \peq  + \frac 1 n \sum_{i = 1}^n \nabla m_g(\bm X_i)  \{ \wh \nabla m_g(\bm X_i) - \nabla m_g(\bm X_i)\}^\top                                                                                       \\
                      & \peq  + \frac 1 n \sum_{i = 1}^n \{ \wh \nabla m_g(\bm X_i) - \nabla m_g(\bm X_i)\} \nabla  m_g(\bm X_i)^\top                                                                                       \\
                      & \peq + \frac 1 n \sum_{i = 1}^n \{ \wh \nabla m_g(\bm X_i) - \nabla m_g(\bm X_i)\}\{ \wh \nabla m_g(\bm X_i) - \nabla m_g(\bm X_i)\}^\top                                                           \\
                      & = D_1 + D_2 + D_3 + D_4                                                                                                                   = D_1 + D_2 + D_3 + O_p(\sigma_{n, t}^2 + r_{n, d, t}^2).
\end{align*}
$D_1$ is an oracle version of $\wh{\bm \Delta}$ and $D_1 = \bm \Delta_g + O_p(n^{-1/2})$. The other terms come from estimating $\nabla m_g$.

Recall from \autoref{thm:ll_refined}
\begin{align*}
  \sup_{\bm x \in \Dcal_{\bm X}, g \in \Gcal} \biggl\| \wh \nabla m_g(\bm x) -  \nabla m_g(\bm x) - \frac 1 n\sum_{i = 1}^n \bm \psi_n(\bm Y_i, \bm X_i, \bm x) \biggr\| = O_p(\sigma_{n, t}),
\end{align*}
where
\begin{align*}
  \bm \psi_n(\bm Y_i, \bm X_i, \bm x)  = \frac 1 {f_{\bm B_0^\top}(\bm B_0^\top \bm x)}\bm \Gamma^{+}(\bm x) \epsilon_{g}(\bm Y_i, \bm X_i)\bigl\{\bm X_i - \bm \mu(\bm x)\}K_{h_t}({\bm B_0}^\top\xix).
\end{align*}
Then,
\begin{align*}
  D_2 = D_3^\top & =  \frac 1 {n^2} \sum_{i = 1}^n \sum_{j = 1}^n  \nabla  m_g(\bm X_i) \bm \psi_n(\bm Y_j, \bm X_j, \bm X_i)^\top + O_p(\sigma_{n, t}) \\
                 & =  \binom{n}{2}^{-2}\sum_{i = 1}^n \sum_{j = i + 1}^n  \bm A_n(\bm Y_i, \bm X_i, \bm Y_j, \bm X_j) + O_p(\sigma_{n, t} + n^{-1}),
\end{align*}
where
\begin{align*}
  \bm A_n(\bm Y_i, \bm X_i, \bm Y_j, \bm X_j) = \nabla  m_g(\bm X_i) \bm \psi_n(\bm Y_j, \bm X_j, \bm X_i)^\top + \nabla m_g(\bm X_j) \bm \psi_n(\bm Y_i, \bm X_i, \bm X_j)^\top .
\end{align*}
Since $\E\{ \epsilon_{g}(\bm Y_j, \bm X_j) \mid \bm X_j\} = 0$, it holds $\E\{\bm A_n(\bm Y_i, \bm X_i, \bm Y_i, \bm X_j)\} = 0$ and
\begin{align*}
   & \quad \; \E\{\bm A_n(\bm Y_i, \bm X_i, \bm Y_i, \bm X_j)\mid\bm Y_i, \bm X_i\}                                                                                                                                    \\
   & =\E\{\nabla m_g(\bm X_j) \bm \psi_n(\bm Y_i, \bm X_i, \bm X_j)^\top \mid\bm Y_i, \bm X_i\}                                                                                                                        \\
   & = \epsilon(\bm Y_i, \bm X_i)\E\bigl[\nabla m_g(\bm X_j) \bm \Gamma^{+}(\bm X_j)\{\bm X_i - \bm \mu(\bm X_j)\}K_{h_t}\{\bm B_0^\top(\bm X_i - \bm X_j)\}/f_{\bm B_0^\top}(\bm B_0^\top \bm X_j) \mid \bm X_i\bigr] \\
   & =  \epsilon(\bm Y_i, \bm X_i)\bigl[\nabla m_g(\bm X_i)\{\bm X_i - \bm \mu(\bm X_i)\}^\top \bm \Gamma^{+}(\bm X_i)  + O(h_t^2) \bigr]                                                                              \\
   & =: \tilde{\bm A}(\bm Y_i, \bm X_i).
\end{align*}
Computing
$ %\begin{align*}
  \var\bigl\{ \bm A_n(\bm Y_i, \bm X_i, \bm Y_j, \bm X_j) - \tilde{\bm A}_n(\bm Y_i, \bm X_i) - \tilde{\bm A}_n(\bm Y_j, \bm X_j)\bigr\} = O(h_t^{-d}),
$ %\end{align*}
standard results for U-statistics \citep[e.g.,][]{korolyuk2013theory} yield
\begin{align*}
  \binom{n}{2}^{-2}\sum_{i = 1}^n \sum_{j = i + 1}^n  \bm A_n(\bm Y_i, \bm X_i, \bm Y_j, \bm X_j) & = \frac 1 n \sum_{i = 1}^n \tilde{\bm A}(\bm Y_i, \bm X_i) + O_p(n^{-1}h_t^{-d/2}) \\
                                                                                                  & = O_p(n^{-1/2} + r_{n, d}^2).
\end{align*}
We have shown that,
$ %\begin{align*}
  \| \wh{\bm \Delta}_g - \bm \Delta_g \| = O_p(n^{-1/2}  + \sigma_{n, t} + r_{n, d, t}^2).
$ %\end{align*}
Then
\begin{align*}
  \delta_{\wh{\bm B}^{(t)}} & = O_p(n^{-1/2} + \sigma_{n, t} + r_{n, d, t}^2)                                                                                          \\
                            & = O_p(n^{-1/2}   + h_t^3 + h_t^2r_{n, d, t} + h_t\delta_{\bm B^{(t - 1)}} + r_{n, d, t} \delta_{\wh{\bm B}^{(t)}} / h_t + r_{n, d, t}^2) \\
                            & = O_p\{n^{-1/2}   + h_t^3 + r_{n, d, t}^2 +  \delta_{\wh{\bm B}^{(t - 1)}}(h_0 + r_{n, d, \infty} / h_\infty)\},
\end{align*}
which we write as $O_p(a_{n, t} +  \delta_{\wh{\bm B}^{(t - 1)}}b_n)$ for simplicity.
Iteratively substituting this expression for $\delta_{\wh{\bm B}^{(t - 1)}}$ yields
\begin{align*}
  \delta_{\wh{\bm B}^{(t + 1)}} = O_p\biggl\{a_{n, t} + \sum_{k = 1}^{t - 1}a_{n, t - k}b_n^{k} + \delta_{\wh{\bm B}^{(0)}}b_n^t\biggr\}.
\end{align*}
Since $h_t = \max\{\rho h_{t - 1}, h_{\infty}\}$ with $\rho \in (0, 1)$, it holds $a_{n, t - k} = O(a_{n, t})$ for all $k \le k^*$ and fixed $k^* < \infty$. Since also $b_n = o(1)$,
\begin{align*}
  \delta_{\wh{\bm B}^{(t + 1)}} = O_p\biggl\{a_{n, t} + \sum_{k = k^*}^{t - 1}a_{n, t - k}b_n^{k} + \delta_{\wh{\bm B}^{(0)}}b_n^t\biggr\}.
\end{align*}
Furthermore, there is a finite $k^*$ such that $b_n^{k^*} = O(a_{n, t})$. Hence, for  $t \ge k^*$, $\delta_{\wh{\bm B}^{(t + 1)}} = O_p(a_{n, t}\bigr) = O_p(n^{-1/2}   + h_t^4  + r_{n, d, t}^2).$ Iterating until $h_t = h_\infty$ proves our claim. \qed

	\end{appendices}
	% load references

	\bibliography{References}

\end{document}